\documentclass[11pt]{article}
\usepackage{color}
\usepackage{graphicx}
\usepackage{amsmath}
\usepackage{authblk}
\usepackage{times}
\usepackage{subfigure}
\usepackage{cite}
\usepackage{amssymb}
\usepackage{amsthm}
\usepackage{geometry}
\usepackage[margin=10pt, font=normalsize,labelfont=bf,labelsep=quad,justification=raggedright]{caption}
\usepackage{sectsty}

\renewcommand{\baselinestretch}{1.5}

\def\diag{{\rm diag}\,}

\def\Exp{{\mathbb{E}}\,}
\def\tr{{\rm tr}\,}

\def\diag{{\rm diag}\,}
\def\real{{\rm Re}\,}

\def\be{\begin{equation}}
\def\ee{\end{equation}}
\def\ba{\left[\begin{array}}
\def\ea{\end{array}\right]}
\def\bea{\begin{eqnarray}}
\def\eea{\end{eqnarray}}

\newcommand{\mb}[1]{\mathbf{#1}}

\newcommand{\mc}[1]{\mathcal{#1}}
\newcommand{\ol}[1]{\overline{#1}}

\def\ba{{\bf a}}

\def\br{{\bf r}}

\def\bv{{\bf v}}

\newcommand{\subscript}[1]{\ensuremath{_{\textrm{#1}}}}
\def\IC-Relay-TDMA{{Concurrent}\subscript{S$\rightarrow$R}-{IC}\subscript{{R}}}
\def\DSTC-ICRec{{Concurrent}\subscript{S$\rightarrow$R$\rightarrow$D}-{IC}\subscript{{D}}}
\def\TDMA-ICRec{{Concurrent}\subscript{R$\rightarrow$D}-{IC}\subscript{{D}}}

\def\full-TDMA-DSTC{{TDMA}\subscript{S$\rightarrow$R$\rightarrow$D}}
\def\joint{{Concurrent}\subscript{S$\rightarrow$R$\rightarrow$D}}
\def\TDMADFICRec{{Concurrent}\subscript{R$\rightarrow$D}-{D}\subscript{{R}}-{IC}\subscript{{D}}}

\newtheorem{theorem}{\textbf{Theorem}}

\newtheorem{corollary}{\textbf{Corollary}}
\newtheorem{lemma}{\textbf{Lemma}}

\begin{document}
\title{Multi-Source Transmission for Wireless Relay
Networks with Linear Complexity
}
%
\date{}
\author[1]{Liangbin Li}
\author[2]{Yindi Jing}
\author[1]{Hamid Jafarkhani\thanks{Part of this work was presented at IEEE
International Conference on Communications (ICC) 2009.}}
\affil[1]{Center for Pervasive Communications \& Computing,
University of California, Irvine} \affil[2]{University of Alberta}
%
\maketitle
\renewcommand{\baselinestretch}{1.4}
\begin{abstract}
This paper considers transmission schemes in multi-access relay
networks (MARNs) where $J$ single-antenna sources send independent
information to one $N$-antenna destination through one $M$-antenna
relay. For complexity considerations, we propose a linear framework,
where the relay linearly transforms its received signals to generate
the forwarded signals without decoding and the destination uses its
multi-antennas to fully decouple signals from different sources
before decoding, by which the decoding complexity is linear in the
number of sources.
To achieve a high symbol rate, we first propose a scheme called
\DSTC-ICRec in which all sources' information streams are
concurrently transmitted in both the source-relay link and the
relay-destination link. In this scheme, distributed space-time
coding (DSTC) is applied at the relay, which satisfies the linear
constraint. DSTC also allows the destination to conduct the
zero-forcing interference cancellation (IC) scheme originally
proposed for multi-antenna systems to fully decouple signals from
different sources. Our analysis shows that the symbol rate of
\DSTC-ICRec is $1/2$ symbols/source/channel use and the diversity
gain of the scheme is upperbounded
by $M-J+1$. 
To achieve a higher diversity gain, we propose another scheme called
\TDMA-ICRec in which the sources time-share the source-relay link.
The relay coherently combines the signals on its antennas to
maximize the signal-to-noise ratio (SNR) of each source, then
concurrently forwards all sources' information. The destination
performs zero-forcing IC. It is shown through both analysis and
simulation that when {\small$N \ge 2J-1$}, \TDMA-ICRec achieves the
same maximum diversity gain as the full TDMA scheme in which the
information stream from each source is assigned to an orthogonal
channel in both links, but with a higher symbol rate.
\end{abstract}
{\bf\em Index Terms:} Multi-access relay network, distributed
space-time coding, interference cancellation, orthogonal and
quasi-orthogonal designs, cooperative diversity.
\renewcommand{\baselinestretch}{1.6}
\section{Introduction}
Node cooperation improves the reliability and the capacity of
wireless networks. Recently, many cooperative schemes have been
proposed, and their multiplexing and diversity gains are analyzed
\cite{LenemanWornell,DSTC-paper,zz-eg-sc,LaTsWo}. However, most
pioneer works in this area focus on cooperative relay designs
without multi-user interference. It is assumed that there is a
single transmission at a time or orthogonal channels are assigned to
different transmissions, e.g.
\cite{LenemanWornell,DSTC-paper,zz-eg-sc,LaTsWo}. As a general
network has multiple nodes each of which can be a data source,
allocating an orthogonal channel to the information stream of each
source is bandwidth inefficient. Therefore, concurrent transmission
of information streams from multiple sources is desirable in
cooperative networks to improve spectrum efficiency. Some examples
on the design and performance analysis of multi-source transmission
can be found in \cite{L06,O08,K09,Yilmaz-ICC}.

One model on multi-source transmission is the interference relay
network\cite{MoBoNa05}. Multiple parallel communication flows are
supported by a common set of cooperative relays through two hops of
transmission. Each source targets at one distinct destination. Two
schemes using relays to resolve interference were discussed. The
zero-forcing (ZF) relaying scheme designs scalar gain factors at
single-antenna relays to null out interference at undesired
destinations\cite{WitRan04, Wit06, Niu07}. The minimum mean square
error (MMSE) relaying scheme designs scalar gain factors to minimize
interference-plus-noise power at undesired
destinations\cite{BerWit05, Keyi-ICCASP}. However, both relaying
schemes assume that the gain factors are first calculated at one
centralized node having perfect and global channel state information
(CSI), then fed back to the relays. While papers \cite{WitRan04,
Wit06, Niu07, BerWit05, Keyi-ICCASP} discuss the multiplexing gain
and designs of the optimal scalar gain factors, they do not provide
diversity analysis. An interference relay network with multi-antenna
nodes was discussed in \cite{Oyman07}, in which the authors used
maximum-ratio-combining (MRC), ZF, and MMSE relaying schemes and
analyzed the power-bandwidth trade-off of the network.

Another model on multi-source cooperative communication considers
the scenario where several sources target at one multi-antenna
destination with the help of one multi-antenna relay. The network is
called \emph{multi-access relay network}
(MARN)\cite{ICRelay-TDMA-jou}. We use the notation $1_J\times
M_1\times N_1$ to represent the MARN with $J$ single-antenna
sources, one $M$-antenna relay, and one $N$-antenna destination. For
the MARN, the source-relay link is a multi-access channel (MAC) and
the relay-destination link is a point-to-point multiple input
multiple output (MIMO) channel. The MARN is thus essentially a
serial concatenation of the MAC and the MIMO. Both links have the
potential for \emph{multi-source concurrent transmission}, i.e.,
information streams from different sources can be simultaneously
transmitted on the same channel. An intuitive scheme is to allow
information streams from different sources concurrently transmitted
in both links and jointly decode all sources' information at the
relay and the destination. Single source transmission schemes, e.g.,
distributed space time code (DSTC), can be applied straightforwardly
following this idea by treating signals from different sources
jointly as a higher dimension signal vector. It can be shown that
this scheme achieves a symbol rate of 1/2 symbols/source/channel use
and the maximum diversity gain of $M$. However, with such a scheme,
the decoding complexities at the relay and the destination are
exponential in the number of sources, thus may become infeasible for
networks with a large number of sources. For complexity
considerations, we propose a \emph{linear framework} for MARNs. The
relay linearly transforms its received signals to generate the
forwarded signals without decoding. The destination separates
signals from different sources before the ML decoding of each
source's information. The decoding complexity at the destination is
hence linear in the number of sources. To the best of our knowledge,
MARNs with this linear framework have not been explicitly discussed
in the literature. It is noteworthy that this linear framework may
constrain the network optimality in some performance measures.

For single-source two-hop cooperative networks, DSTC can achieve the
maximum diversity gain without any CSI at the relay\cite{DSTC-OD}.
For the multi-source scenario, one can use DSTC at the relay and
assign the information stream of each source an orthogonal channel
in both links. This scheme is denoted as \full-TDMA-DSTC, whose
achievable diversity gain is $M$ for $1_J\times M_1\times N_1$
MARNs\cite{DSTC-mulpaper}. Since interference is avoided in both
links, we call this maximum diversity gain the
\emph{interference-free} (int-free) diversity gain. It provides a
natural upperbound on the diversity gain for any multi-source
transmission scheme that allows concurrent transmission of
information streams from different sources. However, \full-TDMA-DSTC
has low spectrum efficiency when the number of sources is large. In
\cite{ICRelay-TDMA-jou}, we proposed a multi-source transmission
scheme called IC-Relay-TDMA, in which concurrent multi-source
transmission is allowed in the source-relay link only. The relay,
knowing the source-relay channel, performs linear interference
cancellation (IC)\cite{NaSeCa,AlCa,KaJa} to decouple signals from
different sources. In the relay-destination link, the relay forwards
information of different sources to the destination using TDMA. To
adopt the same naming system, this scheme is denoted as
\IC-Relay-TDMA instead in this paper. For a $L_{J}\times M_1\times
N_1$ MARN, \IC-Relay-TDMA achieves the maximum int-free diversity
when $N\le L\left(1-\frac{J-1}{M}\right)$\cite{ICRelay-TDMA-jou}.
For the MARN considered in this paper, i.e., each source has only
one antenna, \IC-Relay-TDMA only achieves a diversity gain of
$M-J+1$, hence cannot achieve the maximum int-free diversity gain.
Also, the TDMA method in the relay-destination link limits the
symbol rate of the network.

The \IC-Relay-TDMA scheme uses the relay to remove interference from
different sources. For the considered MARN, the multi-antenna
destination also has the capability of IC. In this paper, we propose
two schemes in which IC is conducted at the destination rather than
the relay. This is desirable for networks with powerful destinations
such as the uplink of cellular systems. The first protocol allows
information streams from different sources simultaneously
transmitted in both links. The relay conducts DSTC to linearly
transform its received signals without decoding. The destination
performs IC to separate signals from different sources. Hence, this
protocol is called \emph{\DSTC-ICRec}. For the second protocol, the
sources time-share the source-relay link. The relay obtains
soft-estimates of the symbols from each source by MRC, encodes
soft-estimates of each source by one DSTC, then concurrently
forwards all sources' DSTCs. The destination performs IC to decouple
signals from different sources. Since information streams of
different sources are simultaneously transmitted in the
relay-destination link only, we call this protocol
\emph{\TDMA-ICRec}. A brief comparison of the proposed protocols
with \full-TDMA-DSTC and \IC-Relay-TDMA in symbol rate, diversity
gain, and CSI requirements is illustrated in Table \ref{table-comp}.
Contributions of the proposed protocols are summarized as follows.
\begin{enumerate}
  \item The proposed protocols fit the linear framework: linear processing without decoding at the relay and linear decoding complexity in the number of sources at the destination. Furthermore, they are applicable to the interference relay network.
  \item CSI feedback, which is necessary for ZF and MMSE relaying schemes\cite{WitRan04,Wit06,BerWit05, Keyi-ICCASP}, is not required for
  the
  protocols proposed in this paper.
  \item We perform rigorous analysis on the diversity gain of the proposed protocols, which to the best of our knowledge, is not provided for related work on multi-source
  cooperative networks.
  \item \DSTC-ICRec achieves a symbol rate of $1/2$ symbols/source/channel use, the highest among the linear schemes in Table \ref{table-comp}. Since the symbol rate of each source is independent of the number of sources, the throughput of the network grows linearly with the number of sources without increasing the bandwidth. With rigorous analysis, the diversity gain is shown to be upperbounded by $M-J+1$.
  \item \TDMA-ICRec achieves a symbol rate of $\frac{1}{J+1}$ symbols/source/channel use in conjunction with a diversity gain of $\min\{M, \lfloor\frac{M}{J}\rfloor(N-J+1)\}$~($\lfloor x\rfloor$ denotes the maximum
integer not greater than $x$). When $N\ge 2J-1$, \TDMA-ICRec
achieves the maximum int-free diversity gain. Compared with
\full-TDMA-DSTC, it has a higher symbol rate with no penalty on the
diversity gain for networks satisfying $N\ge 2J-1$. Compared with
\IC-Relay-TDMA\cite{ICRelay-TDMA-jou}, it has the same symbol rate
but has advantage in the diversity gain for the $1_J\times M_1\times
N_1$ MARN.
\end{enumerate}

The rest of the paper is organized as follows. Section
\ref{sec-model} introduces the network model. Section
\ref{sec-DSTC-ICRec} presents \DSTC-ICRec and analyzes its diversity
gain. In Section \ref{sec-TDMA-ICRec}, \TDMA-ICRec is proposed and
its performance is studied. Section \ref{sec-Simulation} provides
the numerical results. Conclusions are given in Section
\ref{sec-Conclusion}. Involved proofs are presented in appendices.

Notation: For a matrix $\mb{A}$, denote its $(i,j)$th entry as
$a_{ij}$. $\mb{A}^t$, $\mb{A}^*$, and $\overline{\mb{A}}$ are the
transpose, Hermitian, and conjugate of $\mb{A}$, respectively.
$\|\mb{A}\|$ is the Frobenius norm of $\mb{A}$. For two matrices
$\mb{A}$ and $\mb{B}$ of the same dimension, $\mb{A}\succ\mb{B}$
means that $\mb{A}-\mb{B}$ is positive definite. $\mb{I}_n$ is the
$n\times n$ identity matrix. $\mb{0}_{mn}$ is the $m\times n$ matrix
of all zeros. When $m=n$, $\mb{0}_{nn}$ is simplified as $\mb{0}_n$.
$f(x)=o(x)$ means $\underset{x\rightarrow
0^+}{\lim}\frac{f(x)}{x}=0$. $\Exp [x]$ denotes the expected value
of the random variable $x$.

\section{Network Model}
\label{sec-model} Consider a MARN with $J$ single-antenna sources,
one $M$-antenna relay, and one $N$-antenna destination, where there
is no direct connection from the sources to the destination. This
MARN is denoted as a $1_J\times M_1\times N_1$ MARN. We further
assume that both the numbers of relay antennas and destination
antennas are no less than the number of sources, i.e., $J\le
\min\{M, N\}$. This condition is to guarantee full IC at the
destination, the details of which will be shown later. The condition
can be realized by user admission control in the upper-layers. We
assume that both the relay and the destination know the value of
$J$.

Denote the channel coefficient from Source $j$ $(j=1,\ldots, J)$ to
the $i$-th $(i=1,\ldots, M)$ relay antenna as $f_{i}^{(j)}$, and the
channel coefficient from the $i$-th relay antenna to the $n$-th
$(n=1,\ldots, N)$ destination antenna as $g_{in}$. Assume that all
channel coefficients are i.i.d.~circularly symmetric $\mc{CN}(0,1)$
distributed. In addition, we assume a block-fading model with
coherent interval $T$. The noises at each relay antenna and
destination antenna are modeled as additive white Gaussian noise
(AWGN) with zero mean and unit power.
Throughout the paper, we assume global and perfect CSI at the
destination. The CSI requirement at the relay depends on the scheme.
In Section \ref{sec-DSTC-ICRec}, the proposed protocol does not need
any CSI at the relay; in Section \ref{sec-TDMA-ICRec}, the relay
needs only backward CSI, i.e., the channel information from all
sources to itself. The required backward CSI can be acquired by
training\cite{DSTC-mulpaper, SunJing2010}. No CSI feedback is
required for either protocol. To focus on the diversity gain
performance, we assume that all sources and the relay have the same
average power constraint. Further, all nodes are assumed to be
perfectly synchronized at the symbol level.

\section{The Protocol of \DSTC-ICRec}
\label{sec-DSTC-ICRec}

In this section, we propose a protocol that allows concurrent
transmission of information streams from different sources in both
the source-relay link and the relay-destination link to achieve the
symbol rate of $1/2$ symbols/source/channel use. The protocol is
thus called \DSTC-ICRec. Based on the linear framework introduced in
Section 1, we need to design the linear signal processing at the
relay and the destination. Since DSTC requires only a linear
transformation at the relay and achieves the maximum diversity gain
in single-source relay networks\cite{DSTC-paper,DSTC-OD}, we propose
to use DSTC for the MARN to gain protection against channel fading.
At the destination, the IC method\cite{NaSeCa,AlCa,KaJa}, originally
proposed for multi-user MAC to decouple interfering
signals\cite{KaJa-2}, is used to separate information of different
sources. In Subsection \ref{subsec-protocol}, we describe the
details of the protocol. Subsection \ref{subsec-divana} provides the
diversity gain analysis. Subsection \ref{subsec-dis} contains the
discussion on the condition for full IC at the destination and the
symbol rate.

\subsection{Protocol Description}\label{subsec-protocol}

We first describe \DSTC-ICRec in the $1_2\times 2_1\times N_1$ MARN
with two single-antenna sources, one double-antenna relay, and one
$N$-antenna destination; then consider the $1_2\times 4_1\times N_1$
MARN followed by its generalization to $1_J\times M_1\times N_1$
MARNs.

\subsubsection{\DSTC-ICRec for the $1_2\times 2_1\times N_1$
MARN}\label{subsec-motivate} The protocol of \DSTC-ICRec consists of
two steps as shown in Fig.~\ref{fig-DSTCICRec-block}. During the
first step, each source collects two symbols $s_1^{(j)}$ and
$s_2^{(j)}$ independently and uniformly from the constellation
$\mc{S}$. Source $j$ transmits a vector of two symbols
$\mb{x}^{(j)}=\left[s_1^{(j)}\ s_2^{(j)}\right]^t$ and both sources
transmit concurrently. The received signal vector at the $i$-th
relay antenna can be expressed as
\be\label{eq-transmission1}\br_{i}=\sqrt{P}\mb{x}^{(1)}f_{i}^{(1)}+\sqrt{P}\mb{x}^{(2)}f_{i}^{(2)}+\bv_{i},\ee
where $\mb{v}_{i}$ denotes the $2\times 1$ AWGN vector at the $i$-th
relay antenna. The relay uses Alamouti DSTC\cite{DSTC-OD} to
generate its output signal vector at the $i$-th antenna as,
\begin{eqnarray}\label{eq-DSTC}
\mb{t}_i&=&\sqrt{\frac{P}{4P+2}}\left(\mb{A}_i\mb{r}_i+\mb{B}_i\ol{\mb{r}_i}\right),\
i=1, 2,
\end{eqnarray}
where $\sqrt{\frac{P}{4P+2}}$ normalizes the average power at the
relay to $P$ and $\mb{A}_i$ and $\mb{B}_i$ are the $2\times 2$
encoding matrices based on Alamouti design \cite{hj}:
\be\label{eq-Alamouti}\mb{A}_1=\mb{I}_{2}, \mb{B}_2=\left[\begin{array}{rr}0&-1\\
1&0\end{array}\right], \mb{A}_2=\mb{B}_1=\mb{0}_2. \ee From
\eqref{eq-DSTC}, the output vector $\mb{t}_i$ is a linear
transformation of the input vector $\mb{r}_i$, i.e., the relay
signal processing is a linear transformation. Since this linear
transformation is independent of the channels, the relay does not
need any CSI.

During the second step, the relay transmits $\mb{t}_i$ from its
$i$-th antenna, and $\mb{t}_1$ and $\mb{t}_2$ are concurrently
transmitted. Denote the sampled signal at the $n$-th antenna of the
destination and time slot $\tau$ as $x_{\tau n}$. Using the special
structure of the Alamouti design, an equivalent system can be
obtained as
{\small\setlength{\arraycolsep}{0pt}\begin{eqnarray}\label{eq-syseq}
\underset{\tilde{\mb{x}}_n}{\underbrace{\left[\begin{array}{c}x_{1n}\\
\ol{x_{2n}}\end{array}\right]}}
=\sqrt{\frac{P^2}{4P+2}}\sum_{j=1:2}\underset{\mb{H}^{(j)}_n}{\underbrace{\left[\begin{array}{cc}f_{1}^{(j)}g_{1n} &-\ol{f_{2}^{(j)}}g_{2n}\\ f_{2}^{(j)}\ol{g_{2n}}\ &\ol{f_{1}^{(j)}g_{1n}}\end{array}\right]}}\left[\begin{array}{c}s_1^{(j)}\\
\ol{s_2^{(j)}}\end{array}\right]+\underset{\mb{u}_n}{\underbrace{\sqrt{\frac{P}{4P+2}}\left(\left[\begin{array}{c}g_{1n}v_{11}\\
\ol{g_{1n}}v_{21}\end{array}\right]+\left[\begin{array}{c}-g_{2n}\ol{v_{22}}\\
\ol{g_{2n}v_{12}}\end{array}\right]\right)+\left[\begin{array}{c}w_{1n}\\
\ol{w_{2n}}\end{array}\right]}},\end{eqnarray}} \hspace{-4pt}where
$w_{\tau n}$ denotes the AWGN at the $n$-th destination antenna and
time slot $\tau$ and $\mb{u}_n$ denotes the equivalent noise vector
at the $n$-th antenna of the destination. The $2\times 2$ equivalent
channel matrix $\mb{H}_n^{(j)}$ for Source $j$ has Alamouti
structure, i.e.,
{\small$\mb{H}^{(j)*}_n\mb{H}^{(j)}_n=\left(\left|f_1^{(j)}g_{1n}\right|^2+\left|f_2^{(j)}g_{2n}\right|^2\right)\mb{I}_2$}.

Note that the equivalent system equation in \eqref{eq-syseq} is
similar to that of a MAC with two double-antenna users except that
the noise vector is correlated. Using the IC techniques proposed for
MAC in \cite{AlCa}, the destination can fully decouple the
information streams from different sources and separately decode the
information of each source. Without loss of generality, we discuss
how the destination decodes the information of Source 1. To cancel
the symbols of Source 2, the destination calculates
{\small$\hat{\mb{x}}_n=\frac{2\mb{H}_n^{(2)*}}{\|\mb{H}_n^{(2)}\|^2}\tilde{\mb{x}}_n-\frac{2\mb{H}_N^{(2)*}}{\|\mb{H}_N^{(2)}\|^2}\tilde{\mb{x}}_N$}
for $n=1, \ldots, N-1$. Define
$\tilde{\mb{x}}=[\tilde{\mb{x}}_1^*,\ldots,\tilde{\mb{x}}_N^*]^*$,
which is a $2N\times 1$ vector, and
$\mc{X}=[\hat{\mb{x}}_1^*,\ldots,\hat{\mb{x}}_{N-1}^*]^*$, which is
a $(2N-2)\times 1$ vector. The IC process can be represented in a
matrix form as
\begin{eqnarray}\label{eq-remaining}
\mc{X}=\mb{B}\tilde{\mb{x}}=\sqrt{\frac{P^2}{4P+2}}\mb{B}\mb{H}_1\left[\begin{array}{c} s_{1}^{(1)} \\
\overline{s_{2}^{(1)}}\end{array}\right]+\underset{\mb{n}}{\underbrace{\mb{B}\mb{u}}},
\end{eqnarray}
where the $(2N-2) \times 2N$ matrix $\mb{B}$ is the IC matrix, the
$2N\times 2$ matrix $\mb{H}_1$ denotes the equivalent channel matrix
for Source 1, and the $(2N-2)\times 1$ vector $\mb{n}$ denotes the
remaining equivalent noise vector. $\mb{B}$, $\mb{H}_1$, and
$\mb{u}$ are given as
{\small\setlength{\arraycolsep}{2pt}\begin{equation}\label{eq-ICmatrix}
\mb{B}=\left[\begin{array}{ccccc}\frac{2\mb{H}^{(2)*}_{1}}{\|\mb{H}^{(2)}_{1}\|^2}&\mb{0}_2&\cdots&\mb{0}_2&-\frac{2\mb{H}^{(2)*}_{N}}{\|\mb{H}^{(2)}_{N}\|^2}\\
\mb{0}_2&\frac{2\mb{H}_{2}^{(2)*}}{\|\mb{H}_{2}^{(2)}\|^2}&\cdots&\mb{0}_2&-\frac{2\mb{H}_{N}^{(2)*}}{\|\mb{H}_{N}^{(2)}\|^2}\\
\vdots&\vdots&\ddots&\vdots&\vdots\\
\mb{0}_2&\cdots&\cdots&\frac{2\mb{H}_{N-1}^{(2)*}}{\|\mb{H}_{N-1}^{(2)}\|^2}&-\frac{2\mb{H}_{N}^{(2)*}}{\|\mb{H}_{N}^{(2)}\|^2}\end{array}\right],
\mb{H}_1=\left[\begin{array}{c}\mb{H}^{(1)}_{1}\\\mb{H}^{(1)}_{2}\\\vdots\\
\mb{H}^{(1)}_{N}\end{array}\right],\
\mb{u}=\left[\begin{array}{c}\mb{u}_1\\\mb{u}_2\\\vdots\\\mb{u}_{N}\end{array}\right].
\end{equation}}

It can be shown that the equivalent noise vector $\mb{n}$ is
Gaussian but not white. With straightforward calculation, the
$(2N-2)\times (2N-2)$ covariance matrix of $\mb{n}$ can be obtained
as
\be\label{eq-noisecov}\mb{R}_\mb{n}=\frac{P}{4P+2}\mb{B}\tilde{\mb{G}}\tilde{\mb{G}}^*\mb{B}^*+\mb{B}\mb{B}^*,\ee
where \setlength{\arraycolsep}{1pt}{\be\label{eq-G}\tilde{\mb{G}}=
\left[\tilde{\mb{G}}_1^t\ \cdots\ \tilde{\mb{G}}_N^t\right]^t,\
\tilde{\mb{G}}_n=
\left[\begin{array}{cccc}g_{1n}&0&g_{2n}&0\\0&\ol{g_{1n}}&0&
\ol{g_{2n}}\end{array}\right].\ee} Based on \eqref{eq-remaining},
Source 1's information can be recovered using the maximum-likelihood
(ML) decoding rule \be\label{eq-DSTCICRec-ML}
\arg\underset{s_1^{(1)},
s_2^{(1)}}{\min}\left(\mc{X}-\sqrt{\frac{P^2}{4P+2}}\mb{B}\mb{H}_1\left[\begin{array}{c}s_1^{(1)}\\\ol{s_2^{(1)}}\end{array}\right]\right)^*\mb{R}_\mb{n}^{-1}\left(\mc{X}-\sqrt{\frac{P^2}{4P+2}}\mb{B}\mb{H}_1\left[\begin{array}{c}s_1^{(1)}\\\ol{s_2^{(1)}}\end{array}\right]\right).
\ee

Next, we show that the ML decoding in \eqref{eq-DSTCICRec-ML} can be
decoupled into two symbol-wise ML decodings. It suffices to show
that $\mb{H}_1^*\mb{B}^*\mb{R}_\mb{n}^{-1}\mb{B}\mb{H}_1$ is a
diagonal matrix. Notice that Alamouti structure\cite{hj} is closed
under matrix addition, matrix multiplication, and scalar
multiplication. Since the $2\times 2$ submatrices in $\mb{B}$,
$\mb{H}_1$, and $\tilde{\mb{G}}$ have Alamouti structure from
\eqref{eq-ICmatrix} and \eqref{eq-G}, the matrix
$\mb{H}_1^*\mb{B}^*\mb{R}_\mb{n}^{-1}\mb{B}\mb{H}_1$ also has
Alamouti structure in addition to being Hermitian. Generally, it can
be shown that any Hermitian Alamouti matrix is diagonal with equal
diagonal entries. Therefore,
$\mb{H}_1^*\mb{B}^*\mb{R}_\mb{n}^{-1}\mb{B}\mb{H}_1$ is diagonal.
The ML decoding in \eqref{eq-DSTCICRec-ML} can be decomposed to two
procedures of symbol-wise decoding as
\begin{eqnarray*} \arg\underset{s_1^{(1)}}{\max}\
2 \real
\left(\mb{h}_1^{*}\mb{B}^*\mb{R}_\mb{n}^{-1}\mc{X}s_1^{(1)}\right)-\sqrt{\frac{P^2}{4P+2}}\mb{h}_1^{*}\mb{B}^*\mb{R}_\mb{n}^{-1}\mb{B}\mb{h}_1\left|s_1^{(1)}\right|^2,\\
 \arg\underset{s_2^{(1)}}{\max}\
2 \real
\left(\mb{h}_2^{*}\mb{B}^*\mb{R}_\mb{n}^{-1}\mc{X}\ol{s_2^{(1)}}\right)-\sqrt{\frac{P^2}{4P+2}}\mb{h}_2^{*}\mb{B}^*\mb{R}_\mb{n}^{-1}\mb{B}\mb{h}_2\left|s_2^{(1)}\right|^2,
\end{eqnarray*}
where $\mb{h}_i$ denotes the $i$-th column of $\mb{H}_1$. Similarly,
the destination can cancel the symbols of Source 1 and decode the
information of Source 2. Four procedures of symbol-wise ML decoding
are needed in total to decode both sources' information.

\subsubsection{\DSTC-ICRec for the $1_J\times 4_1\times N_1$ MARN}
This subsection describes \DSTC-ICRec in the MARN with four relay
antennas and $J$ sources. During the first step, Source $j$
transmits a $4\times 1$ vector consisting of four symbols,
i.e.,$\left[s_1^{(j)}\ s_2^{(j)} \ s_3^{(j)}\ s_4^{(j)} \right]^t$,
and all sources transmit concurrently. The $i$-th relay antenna
receives a $4\times 1$ vector $\mb{r}_i$. The relay performs DSTC
with quasi-orthogonal design\cite{DSTC-OD}. The $4\times 1$
forwarded vector $\mb{t}_i$ is generated as
$\mb{t}_i=c\left(\mb{A}_i\mb{r}_i+\mb{B}_i\ol{\mb{r}_i}\right)$,
where $c=\sqrt{\frac{P}{4(JP+1)}}$ is to constrain the power of the
relay to $P$; $\mb{A}_i$ and $\mb{B}_i$ are DSTC encoding matrices
with quasi-orthogonal design\cite{DSTC-OD, hj}:
{\small\setlength{\arraycolsep}{1.5pt}\begin{eqnarray}
\mb{A}_1&=&\mb{I}_4,\
\mb{A}_4=\left[\begin{array}{rrrr}0&0&0&1\\0&0&-1&0\\0&-1&0&0\\1&0&0&0\end{array}\right],\
\mb{B}_2=\left[\begin{array}{rrrr}0&-1&0&0\\1&0&0&0\\0&0&0&-1\\0&0&1&0\end{array}\right],\
\mb{B}_3=\left[\begin{array}{rrrr}0&0&-1&0\\0&0&0&-1\\1&0&0&0\\0&1&0&0\end{array}\right],\
\mb{A}_2=\mb{A}_3=\mb{B}_1=\mb{B}_4=\mb{0}_4.\label{eq-quasides}
\end{eqnarray}}
During the second step, the relay concurrently forwards $\mb{t}_i$.
The received signal at the $n$-th destination antenna can be written
as \be \nonumber \left[\begin{array}{c}x_{1n}\\
x_{2n}\\ x_{3n} \\ x_{4n}\end{array}\right]=\sqrt{P}c \sum_{j=1:J}
\underset{\mb{S}^{(j)}}{\underbrace{\left[\begin{array}{rrrr}s_1^{(j)}&
-\ol{s_2^{(j)}}& -\ol{s_3^{(j)}}& s_4^{(j)}\\ s_2^{(j)}&
\ol{s_1^{(j)}}& \ol{s_4^{(j)}}& -s_3^{(j)}\\ s_3^{(j)}&
-\ol{s_4^{(j)}}& \ol{s_1^{(j)}}& -s_2^{(j)}\\ s_4^{(j)}&
\ol{s_3^{(j)}}& \ol{s_2^{(j)}}& s_1^{(j)}
\end{array}\right]}} \left[\begin{array}{c} f_1^{(j)}g_{1n}\\
\ol{f_2^{(j)}}g_{2n}\\ \ol{f_3^{(j)}}g_{3n} \\
f_4^{(j)}g_{4n}\end{array}\right]+c\left[\begin{array}{rrrr}v_{11}&
-\ol{v_{22}}& -\ol{v_{33}}& v_{44}\\ v_{21}& \ol{v_{12}}&
-\ol{v_{43}}& v_{34}\\ v_{31}& -\ol{v_{42}}& -\ol{v_{13}}& -v_{24}\\
v_{41}& -\ol{v_{32}}& -\ol{v_{23}}& v_{14}
\end{array}\right] \left[\begin{array}{c} g_{1n}\\
g_{2n}\\ g_{3n} \\
g_{4n}\end{array}\right]+\left[\begin{array}{c} w_{1n}\\
w_{2n}\\ w_{3n} \\
w_{4n}\end{array}\right], \ee where $x_{\tau n}$ denotes the sampled
signal at the $n$-th destination antenna and time slot $\tau$. It
can be observed that $\mb{S}^{(j)}$ has quasi-orthogonal structure
due to the DSTC at the relay. Using the IC techniques in
\cite{KaJa}, we can break the system into two equivalent Alamouti
systems as {\small\setlength{\arraycolsep}{4pt}
\begin{eqnarray}\label{eq-alasys-p}
{\left[\begin{array}{c}x_{1n}+x_{4n}\\\ol{x_{2n}}-\ol{x_{3n}}\end{array}\right]}&=&\sqrt{P}c\sum_{j=1:J}{\left[\begin{array}{cc}f^{(j)}_{1}g_{1n}+f^{(j)}_{4}g_{4n}&\ol{f^{(j)}_{2}}g_{2n}-\ol{f^{(j)}_{3}}g_{3n}\\
f^{(j)}_{2}\ol{g_{2n}}-f^{(j)}_{3}\ol{g_{3n}}&-\ol{f^{(j)}_{1}}\ol{g_{1n}}-\ol{f^{(j)}_{4}}\ol{g_{4n}}\end{array}\right]}{\left[\begin{array}{c}s_{1}^{(j)}+s_{4}^{(j)}\\\ol{s_3^{(j)}}-\ol{s_2^{(j)}}\end{array}\right]}+\mb{u}_n^+\\
\label{eq-alasys-n}
{\left[\begin{array}{c}x_{1n}-x_{4n}\\\ol{x_{2n}}+\ol{x_{3n}}\end{array}\right]}&=&\sqrt{P}c\sum_{j=1:J}{\left[\begin{array}{cc}f^{(j)}_{1}g_{1n}-f^{(j)}_{4}g_{4n}&\ol{f^{(j)}_{2}}g_{2n}+\ol{f^{(j)}_{3}}g_{3n}\\
f^{(j)}_{2}\ol{g_{2n}}+f^{(j)}_{3}\ol{g_{3n}}&-\ol{f^{(j)}_{1}}\ol{g_{1n}}+\ol{f^{(j)}_{4}}\ol{g_{4n}}\end{array}\right]}{\left[\begin{array}{c}s_{1}^{(j)}-s_{4}^{(j)}\\-\ol{s_3^{(j)}}-\ol{s_2^{(j)}}\end{array}\right]}+\mb{u}_n^-,
\end{eqnarray}}
\hspace{-3pt}where $\mb{u}_n^+$ and $\mb{u}_n^-$ denote the
equivalent noise vector for each system. They have the following
expressions: {\small\setlength{\arraycolsep}{1pt}
\begin{eqnarray*}
\mb{u}_n^+&=&c\left(\left[\begin{array}{c}(v_{11}+v_{41})g_{1n}\\(\ol{v_{21}}-\ol{v_{31}})\ol{g_{1n}}\end{array}\right]+\left[\begin{array}{c}(-\ol{v_{22}}+\ol{v_{32}})g_{2n}\\(v_{12}+v_{42})\ol{g_{2n}}\end{array}\right]+\left[\begin{array}{c}(-\ol{v_{33}}+v_{23})g_{3n}\\(-v_{43}-\ol{v_{13}})\ol{g_{3n}}\end{array}\right]+\left[\begin{array}{c}(v_{44}+v_{14})g_{4n}\\(-\ol{v_{34}}+\ol{v_{24}})\ol{g_{4n}}\end{array}\right]\right)+\left[\begin{array}{c}w_{1n}+w_{4n}\\\ol{w_{2n}}-\ol{w_{3n}}\end{array}\right]\\
\mb{u}_n^-&=&c\left(\left[\begin{array}{c}(v_{11}-v_{41})g_{1n}\\(\ol{v_{21}}+\ol{v_{31}})\ol{g_{1n}}\end{array}\right]+\left[\begin{array}{c}(-\ol{v_{22}}-\ol{v_{32}})g_{2n}\\(v_{12}-v_{42})\ol{g_{2n}}\end{array}\right]+\left[\begin{array}{c}(-\ol{v_{33}}-v_{23})g_{3n}\\(-v_{43}+\ol{v_{13}})\ol{g_{3n}}\end{array}\right]+\left[\begin{array}{c}(v_{44}-v_{14})g_{4n}\\(-\ol{v_{34}}-\ol{v_{24}})\ol{g_{4n}}\end{array}\right]\right)+\left[\begin{array}{c}w_{1n}-w_{4n}\\\ol{w_{2n}}+\ol{w_{3n}}\end{array}\right].
\end{eqnarray*}}

The destination uses $J-1$ antennas to cancel the symbols of $J-1$
interfering sources by the multi-user IC technique in \cite{AlCa}
for each Alamouti system in \eqref{eq-alasys-p} and
\eqref{eq-alasys-n}, thus decouple information streams from
different sources. The destination then recovers information of each
source separately using the ML decoding. Since the detailed formulas
can be found in \cite{KaJa}, we do not repeat the IC procedure here.

\subsubsection{\DSTC-ICRec for $1_J\times M_1\times N_1$ MARNs}
To use the protocol in MARNs with a general $M$, each source
transmits a vector of $2^n$ symbols, with $2^n$ the minimum number
that is no less than $M$. The relay designs the DSTC using the first
$M$ columns of a $2^n\times 2^n$ quasi-orthogonal space-time block
code (STBC) with ABBA structure\cite{QOD,TBH00}. The destination
separates the system into $2^{n-1}$ Alamouti systems and decouples
the signals from different sources using the IC procedure in
\cite{KaJa, ICRelay-TDMA-jou}. Each source's information can be
decoded separately at the destination. The decoding complexity is
thus linear in the number of sources.

\subsection{Diversity Gain Analysis}\label{subsec-divana}
In this subsection, we analyze the diversity gain of \DSTC-ICRec.
Due to the concatenation of the channels, a direct diversity
analysis from the system equation \eqref{eq-remaining} is
challenging. Instead, we work on an equivalent representation for
the tractability of the analysis. The equivalent representation
captures the effect of the IC at the destination to the first step
of transmission. For \DSTC-ICRec, although the ZF IC procedure is
conducted at the destination, there is a virtual ZF at the relay and
a dimension reduction filtering at the destination. We first derive
this system representation in Subsection \ref{subsec-equivalent},
then prove an upperbound of the diversity gain on \DSTC-ICRec in
Subsection \ref{subsec-divupperbound}.

\subsubsection{An equivalent representation for \DSTC-ICRec in $1_J\times M_1\times N_1$
MARNs}\label{subsec-equivalent}  Since the network parameters and
the processing at the relay and the destination are statistically
equivalent for all sources, the diversity gains of all sources are
identical. Thus, we only focus on Source 1.

For the simplicity of the presentation, we first look at the system
equation of the $1_2\times 2_1\times N_1$ MARN in
\eqref{eq-remaining}. Notice that each entry in the channel matrix
$\mb{B}\mb{H}_1$ is a rational function of the channel coefficients
of both links. Then, the entries are neither independent nor
Gaussian. This complicates the diversity gain analysis. In the
following, we derive an equivalent system representation to decouple
the channel concatenation, which will help the diversity gain
analysis. The system equation in \eqref{eq-remaining} can be
rewritten as
\begin{eqnarray}\label{eq-remaining2}
\mc{X}=\sqrt{\frac{P^2}{4P+2}}\mb{B}\tilde{\mb{G}}\mb{F}^{(1)}\left[\begin{array}{c} s_{1}^{(1)} \\
\overline{s_{2}^{(1)}}\end{array}\right]+\mb{n},
\end{eqnarray}
where $\tilde{\mb{G}}$ is defined in \eqref{eq-G} and
{\small\setlength{\arraycolsep}{2pt}$\mb{F}^{(j)}\triangleq\left[\begin{array}{cccc}f_1^{(j)}&0&0&f_2^{(j)}\\0&\ol{f_1^{(j)}}&\ol{f_2^{(j)}}&0\end{array}\right]^t$}.
Note that the IC matrix $\mb{B}$ zero-forces the channels of Source
2, i.e., $\mb{B}\tilde{\mb{G}}\mb{F}^{(2)}=\mb{0}$. In other words,
$\mb{B\tilde{G}}$ nulls out $\mb{F}^{(2)}$. Then, the rows of
$\mb{B\tilde{G}}$ are in the null space of the column space of
$\mb{F}^{(2)}$. Therefore, the channel matrix in
\eqref{eq-remaining2} is invariant if $\mb{F}^{(1)}$ is first
projected onto the null space of $\mb{F}^{(2)}$, i.e.,
$\mb{B\tilde{G}F}^{(1)}=\mb{B\tilde{G}\Phi F}^{(1)}$, where
$\mb{\Phi}$ is the projection matrix to the null space of
$\mb{F}^{(2)}$, i.e.,
$\mb{\Phi}=\mb{I}_4-\frac{2\mb{F}^{(2)}\mb{F}^{(2)*}}{\tr
(\mb{F}^{(2)}\mb{F}^{(2)*})}$. Thus, \eqref{eq-remaining2} can be
rewritten as
\begin{eqnarray}\label{eq-remaining3}
\mc{X}=\sqrt{\frac{P^2}{4P+2}}\mb{B\tilde{G}\Phi F}^{(1)}\left[\begin{array}{c} s_{1}^{(1)} \\
\overline{s_{2}^{(1)}}\end{array}\right]+\mb{n}.
\end{eqnarray}
This new system equation can be interpreted as follows. Symbols
$s_1^{(1)}$ and $s_2^{(1)}$ are first transmitted through channel
$\mb{F}^{(1)}$ to the relay. Then, ZF operation $\mb{\Phi}$ is
conducted to null out the information of Source 2. After that,
signals are forwarded through channel $\tilde{\mb{G}}$ and the
destination applies a filter $\mb{B}$ to reduce the dimension of the
received signal vector from $2N\times 1$ to $2(N-1)\times 1$. The
virtual ZF at the relay and the dimension reduction at the
destination are due to the ZF IC at the destination. A diagram
illustrating this process is shown in Fig.~\ref{fig-eqsystem}.

For general $J$ and $M$, with the same argument, the ZF at the
destination induces a virtual ZF at the relay followed by a
dimension reduction at the destination. It can be shown that the
dimensions of $\mb{F}^{(1)}, \tilde{\mb{G}}, \mb{B}$ are
$2^{n+1}\times 2^n$, $2^nN\times 2^{n+1}$, and $2^n(N-J+1)\times
2^nN$, respectively, where $2^n$ is the minimum number no less than
$M$. The virtual ZF operation at the relay nulls out the information
from Sources $2$ to $J$. The dimension reduction filter $\mb{B}$
decreases the dimension of the received signal vector from
$2^nN\times 1$ to $2^n(N-J+1)\times 1$. Although the new system
representation looks more complicated than the original one, it
simplifies the diversity gain analysis.

\subsubsection{Diversity gain
upperbound}\label{subsec-divupperbound} Diversity gain is defined as
the negative of the asymptotic slope of the bit error rate (BER)
with respect to the average transmit SNR in the high SNR regime. In
\cite{divthm_report}, it is shown that for a communication system
represented by the equation $\mb{y}=\mb{h}s+\mb{n}$ where $\mb{h}$,
$s$, and $\mb{n}$ are the channel vector, scalar symbol, and noise
vector, respectively, diversity gain can be calculated using the
outage probability of the instantaneous normalized receive SNR
$\gamma$ as \be\label{eq-div} d=\lim_{\epsilon\rightarrow
0^+}\frac{\log P(\gamma<\epsilon)}{\log \epsilon}, \ee where the
instantaneous normalized receive SNR is defined as
$\gamma=\mb{h}^*\mb{R}_\mb{n}^{-1}\mb{h}$ and $\mb{R}_\mb{n}$ is the
covariance matrix of $\mb{n}$. This technique is usually easier than
the direct calculation based on the error rate, thus is used in this
paper to analyze the diversity gain of \DSTC-ICRec. Based on the
equivalent system equation in \eqref{eq-remaining3}, the following
theorem is proved.
\begin{theorem}\label{thm-DSTCICRec1}
In $1_J\times M_1\times N_1$ MARNs, the diversity gain of
\DSTC-ICRec is upperbounded by $M-J+1$.
\end{theorem}
\begin{proof}
See Appendix \ref{ap-thm1}.
\end{proof}

Theorem \ref{thm-DSTCICRec1} can be intuitively explained as
follows. Since the first step transmission is a MAC with an
$M$-antenna receiver and the virtual ZF operation at the relay
requires the use of $J-1$ antennas to null out the information of
$J-1$ sources, the diversity gain achievable after the virtual ZF is
no higher than $M-J+1$. The second step transmission and the
dimension reduction at the destination cannot improve the diversity
gain of the first step. Therefore, the protocol has at most a
diversity gain of $M-J+1$. When $J=2$ and $M=2$, the following
diversity result can be obtained as a special case of Theorem
\ref{thm-DSTCICRec1}.
\begin{corollary}\label{cor-DSTCICRec2}
In the $1_2\times 2_1\times N_1$ MARN, the diversity gain of
\DSTC-ICRec is upperbounded by 1.
\end{corollary}

\subsection{Discussion}\label{subsec-dis}
In this subsection, we discuss the properties of \DSTC-ICRec,
including the condition on the network parameters for full IC, the
symbol rate, and its comparison with existing schemes. Finally, we
present its possible applications in the interference relay network.

First we consider the condition on the network parameters to achieve
full IC. From the IC procedure in \cite{KaJa}, at least $J$ receive
antennas are required to decouple signals from $J$ source. Thus,
\DSTC-ICRec requires the number of destination antennas to be no
less than the number of sources. In addition, a condition on the
number of relay antennas is required. 
To show this, we start with the example in the $1_2\times 2_1\times
N_1$ MARN. From \eqref{eq-remaining3}, the equivalent channel vector
experienced by $s_1^{(1)}$ in the source-relay link is the first
column of $\mb{F}^{(1)}$, i.e., $[f_1^{(1)}\ 0\ 0\ f_2^{(1)}]^t$,
which is a $4\times 1$ vector in a 2-dimension subspace. It is
discussed in Subsection \ref{subsec-equivalent} that the IC
operation at the destination creates a virtual ZF operation at the
relay. Then, the equivalent channel vector at the relay can be
projected onto the null space of the equivalent channel vector of at
most one interfering source. In other words, the virtual ZF
operation at the relay can null out interference from at most one
source and the network can allow at most two sources to transmit
simultaneously. In general $1_J\times M_1\times N_1$ MARNs, the
equivalent channel vector at the relay is a $2^{n+1}\times 1$ vector
in a $M$-dimension subspace. The virtual ZF operation at the relay
can null out at most $M-1$ information streams from interfering
sources. Thus, the number of relay antennas also needs to be no less
than the number of sources. With \DSTC-ICRec, the $1_J\times
M_1\times N_1$ MARN admits at most $\min\{M, N\}$ sources to
concurrently transmit information.

Now we discuss the symbol rate of the scheme. The multi-source
transmission in \DSTC-ICRec improves the spectrum efficiency of the
network. In the first step, each source sends a vector of $T=2^n$
symbols in $T$ channel uses where $2^n$ is the minimum number no
less than $M$. Using DSTC with quasi-orthogonal design at the relay
\cite{DSTC-OD}, another $T$ channel uses are required for the second
step. Overall, $2T$ channel uses are required to send $T$ symbols
from end to end. Thus, the symbol rate is $1/2$
symbols/source/channel use, which is independent of the number of
sources.

Next, we compare the diversity gain, symbol rate, and CSI
requirements at the relay of \DSTC-ICRec with two existing schemes:
\full-TDMA-DSTC and \IC-Relay-TDMA\cite{ICRelay-TDMA-jou}, that also
fit the linear framework. The results are shown in Table
\ref{table-comp}. \DSTC-ICRec achieves a higher symbol rate compared
to \IC-Relay-TDMA and \full-TDMA-DSTC. For \DSTC-ICRec, the
throughput of the network grows linearly with the number of sources;
while for the other two, the throughput of the network is
$\frac{J}{1+J}$ for \IC-Relay-TDMA and 1/2 for \full-TDMA-DSTC,
which has an upperbound when $J$ grows large. However, from Theorem
\ref{thm-DSTCICRec1}, the diversity gain of \DSTC-ICRec is
upperbounded by $M-J+1$, thus is inferior to \full-TDMA-DSTC and no
better than \IC-Relay-TDMA. For \DSTC-ICRec, diversity gain
degradation is necessary to trade for a higher symbol rate.
Regarding CSI, the relay does not need any channel information for
\DSTC-ICRec or \full-TDMA-DSTC, while for \IC-Relay-TDMA, the relay
needs to know the source-relay channels.

Finally, \DSTC-ICRec can be applied in more general network models.
\DSTC-ICRec can be used in MARNs with distributed relay antennas.
From \eqref{eq-DSTC}, the forwarded signal from a relay antenna
$\mb{t}_i$ only depends on its own received signal $\mb{r}_i$. No
cross-talk between relay antennas is needed. Thus, the relay
antennas do not have to be collocated to conduct the scheme. It is
the total number of relay antennas that matters. Furthermore,
\DSTC-ICRec can also be straightforwardly used in the interference
relay network with multi-antenna destinations. The protocol
description shows that each destination can use its multi-antennas
to decouple the information streams of multi-sources and decode the
information of its interest as long as the numbers of destination
antennas and distributed relay antennas are no less than the number
of sources.

\section{The Protocol of \TDMA-ICRec}\label{sec-TDMA-ICRec}
Although \DSTC-ICRec improves the spectrum efficiency of MARNs, it
cannot achieve the maximum int-free diversity. In this section, we
propose another protocol that has the potential of achieving the
same int-free diversity gain but with a higher symbol rate compared
to \full-TDMA-DSTC. In $1_J\times M_1 \times N_1$ MARNs, the
source-relay link has $M$ independent channel paths for each source
and the relay-destination link has $MN$ independent channel paths.
The diversity gain is thus bottlenecked by the source-relay link. We
propose to use TDMA in the source-relay link to achieve the maximum
diversity gain, and in the relay-destination link, concurrent
transmission of information streams from different sources is
designed to improve the symbol rate. We denote this protocol as
\TDMA-ICRec. In Subsection \ref{subsec-general}, we present details
of the protocol. We analyze the diversity gain of the protocol in
Subsection \ref{subsec-divana2} and compare it with other schemes in
Subsection \ref{subsec-comp}.

\subsection{Protocol Description}\label{subsec-general}
Since \TDMA-ICRec uses TDMA in the source-relay link, the main
challenge in the design is to allow concurrent transmission of
multi-sources in the relay-destination link, and decouple the
multiple information streams at the destination. Our proposed
protocol requires the relay to know its channels with all sources,
which can be obtained by training, and does not require CSI
feedback. It fits the linear framework introduced in Section 1. In
what follows, we first explain the protocol for the $1_J\times
(2J)_1\times N_1$ MARN, followed by the $1_J\times (4J)_1\times
(N)_1$ MARN, then extend the design to the general case of
$1_J\times M_1\times N_1$ MARNs.

\subsubsection{\TDMA-ICRec for the $1_J\times (2J)_1\times N_1$ MARN}\label{subsec-subsub}
In this subsection, we describe \TDMA-ICRec for the MARN in which
the number of relay antennas is twice that of the number of sources.
The protocol of \TDMA-ICRec consists of two steps as shown in
Fig.~\ref{fig-TDMAICRec-block}. During the first step, two symbols
randomly selected from one constellation $\mc{S}$ are collected by
Source $j$ to form a vector as $\mb{s}^{(j)}=\left[s_1^{(j)}\
s_2^{(j)}\right]^t$. Source $j$ uses time slots $2j-1$ and $2j$ to
send $\mb{s}^{(j)}$. In other words, sources transmit to the relay
in TDMA. In time slots $(2j-1)$ and $2j$, the $i$-th relay antenna
overhears {\be\nonumber
\underset{\mb{r}^{(j)}_i}{\underbrace{\left[\begin{array}{c}r_{(2j-1)i}\\r_{(2j)i}\end{array}\right]}}=\sqrt{P}f^{(j)}_i\mb{s}^{(j)}+\underset{\mb{v}^{(j)}_i}{\underbrace{\left[\begin{array}{c}v_{(2j-1)i}\\v_{(2j)i}\end{array}\right]}},
i=1,\ldots, 2J,\ j=1, \ldots, J, \ee}where $r_{\tau i}$ and $v_{\tau
i}$ denote the received signal and the AWGN at the $i$-th relay
antenna and time slot $\tau$, respectively. The relay coherently
combines signals at each antenna to maximize the SNR of Source $j$'s
transmission and obtains a soft estimate of $\mb{s}^{(j)}$ as, {\be
\label{eq-relayMRCdb}
\hat{\mb{r}}^{(j)}=\frac{\underset{i=1:2J}{\sum}\overline{f^{(j)}_{i}}\mb{r}_i^{(j)}}{\underset{i=1:2J}{\sum}\left|f_{i}^{(j)}\right|^2}=\sqrt{P}\mb{s}^{(j)}+\underset{\hat{\mb{v}}^{(j)}}{\underbrace{\frac{\underset{i=1:2J}{\sum}\overline{f^{(j)}_{i}}\mb{v}^{(j)}_i}{\underset{i=1:2J}{\sum}\left|f_{i}^{(j)}\right|^2}}},
\ee} \hspace{-3pt}where the $2\times 1$ noise vector
$\hat{\mb{v}}^{(j)}$ has
i.i.d.~$\mc{CN}\left(0,\left(\underset{i=1:2J}{\sum}\left|f_{i}^{(j)}\right|^2\right)^{-1}\right)$
entries. The relay uses Alamouti DSTC\cite{DSTC-OD} to encode the
soft estimate of $\mb{s}^{(j)}$ into
{\small\setlength{\arraycolsep}{2pt}\be\label{eq-Alarelay}\left[\begin{array}{cc}\mb{t}_{(2j-1)}&
\mb{t}_{(2j)}\end{array}\right]=\sqrt{\frac{P}{MP+M}}\left[\begin{array}{cc}\mb{A}_1\hat{\mb{r}}^{(j)}&
\mb{B}_2\ol{\hat{\mb{r}}^{(j)}}\end{array}\right], \ee}
\hspace{-3pt}where $\sqrt{\frac{P}{MP+M}}$ is to constrain the
average relay power to $P$; the encoding matrices $\mb{A}_1$ and
$\mb{B}_2$ are given in \eqref{eq-Alamouti}. From
\eqref{eq-relayMRCdb} and \eqref{eq-Alarelay}, the relay generates
the signal $\mb{t}_i$ by a linear transformation from its received
signal $\mb{r}_i^{(j)}$.

During the second step, the relay forwards the $2\times 1$ vector
$\mb{t}_i$ using its $i$-th antenna. All relay antennas transmit
simultaneously to realize concurrent transmissions of all sources'
information streams. From \eqref{eq-Alarelay}, we can see that each
source is assigned two antennas and $J$ Alamouti DSTCs are
concurrently transmitted to the destination. Denote $y_{\tau n}$ as
the received signal at time slot $\tau$ and the $n$-th antenna at
the destination. An equivalent system can be obtained as
{\setlength{\arraycolsep}{2pt}
\be\label{eq-recsigdb}\underset{\tilde{\mb{y}}_n}{\underbrace{\left[\begin{array}{c}y_{1n}\\\ol{y_{2n}}\end{array}\right]}}=\sqrt{\frac{P^2}{MP+M}}\sum_{j=1:J}\underset{\mb{G}_{jn}}{\underbrace{\left[\begin{array}{cc}g_{(2j-1)n}&-g_{(2j)n}\\\ol{g_{(2j)n}}&\ol{g_{(2j-1)n}}\end{array}\right]}}\left[\begin{array}{c}s^{(j)}_1\\\ol{s^{(j)}_2}\end{array}\right]+\sqrt{\frac{P}{MP+M}}\sum_{j=1:J}\mb{G}_{jn}\underset{\tilde{\mb{v}}^{(j)}}{\underbrace{\left[\begin{array}{c}\hat{v}^{(j)}_1\\\ol{\hat{v}_2^{(j)}}\end{array}\right]}}+\underset{\tilde{\mb{w}}_n}{\underbrace{\left[\begin{array}{c}w_{1n}\\\ol{w_{2n}}\end{array}\right]}},\ee
}where $\hat{v}_i^{(j)}$ is the $i$-th entry of $\hat{\mb{v}}^{(j)}$
in \eqref{eq-relayMRCdb}. Note that the $2\times 2$ equivalent
channel matrix $\mb{G}_{jn}$ has Alamouti structure. The equivalent
system equation in \eqref{eq-recsigdb} is similar to that of a
multi-user multi-antenna MAC system except that the equivalent
noises are not white. By applying the multi-user IC schemes in
\cite{AlCa}, the destination can iteratively cancel the symbols of
$J-1$ interfering sources using signals at any $J-1$ antennas. For
full IC, $N\ge J$ is required. Here, we provide a compact matrix
representation of this algorithm, which is not provided in
\cite{AlCa, KaJa}, because the resulting equations are needed for
the diversity analysis. Without loss of generality, we show how the
destination cancels the information of Sources 2 to $J$ and obtains
int-free observations of Source 1 in $J-1$ iterations.

Stack $\tilde{\mb{y}}_n$ to obtain
$\tilde{\mb{y}}=[\tilde{\mb{y}}_1^*,\ldots,\tilde{\mb{y}}_N^*]^*$
and let $\mb{G}_j=[\mb{G}_{j1}^*\ \ldots \mb{G}_{jN}^*]^*$ for $j=1,
\ldots, J$.  The iterative process is described as
follows:{\setlength{\arraycolsep}{2pt}
\begin{itemize}
  \item
  \textbf{Initialization}: $\mc{G}(0)=\left[\begin{array}{ccc}\mb{G}_{1}&\ldots&\mb{G}_{J}\end{array}\right]$,
  $\tilde{\mb{y}}(0)=\tilde{\mb{y}}$.
  \item \textbf{For the $i$-th iteration}: $i=1,\ldots, J-1$
  \begin{enumerate}
    \item Form the $2(N-i)\times 2(N-i+1)$ IC matrix $\mb{B}(i)$ as \\
{\small\begin{eqnarray}\label{eq-ICmatrix2}
  \mb{B}(i)&=&\left[\begin{array}{ccccc}-\frac{2{\mc{G}^*}_{J-i+1,1}(i-1)}{\left\|\mc{G}_{J-i+1,1}(i-1)\right\|^2}&\frac{2{\mc{G}^*}_{J-i+1,2}(i-1)}{\left\|\mc{G}_{J-i+1,2}(i-1)\right\|^2}&\mb{0}_2&\ldots&\mb{0}_2\\
    -\frac{2{\mc{G}^*}_{J-i+1,1}(i-1)}{\left\|\mc{G}_{J-i+1,1}(i-1)\right\|^2}&\mb{0}_2&\frac{2{\mc{G}^*}_{J-i+1,3}(i-1)}{\left\|\mc{G}_{J-i+1,3}(i-1)\right\|^2}&\ldots&\mb{0}_2\\
    \vdots&\vdots&\vdots&\ddots&\vdots\\
    -\frac{2{\mc{G}^*}_{J-i+1,1}(i-1)}{\left\|\mc{G}_{J-i+1,1}(i-1)\right\|^2}&\mb{0}_2&\mb{0}_2&\ldots&\frac{2{\mc{G}^*}_{J-i+1,N-i+1}(i-1)}{\left\|\mc{G}_{J-i+1,N-i+1}(i-1)\right\|^2}\end{array}\right],
   \end{eqnarray}}\\
   where the $2\times 2$ matrix $\mc{G}_{p,q}(i)$ denotes the $(p,q)$-th $2\times 2$ submatrix of
$\mc{G}(i)$.
   \item Cancel the symbols of Source $J-i+1$ by calculating
   $\tilde{\mb{y}}(i)=\mb{B}(i)\tilde{\mb{y}}(i-1)$.
    \item Form the $2(N-i)\times 2J$ remaining equivalent channel matrix $\mc{G}(i)$ as $
    \mc{G}(i)=\mb{B}(i)\mc{G}(i-1).$
  \end{enumerate}
\end{itemize}}

Note that $\tilde{\mb{y}}(i)$ is the $2(N-i)\times 1$ signal vector
after cancelling the information Source $J-i+1$. After $J-1$
iterations, $\tilde{\mb{y}}(J-1)$ only contains the information of
Source 1 and has dimension $2(N-J+1)\times 1$. Let
$\mb{B}=\underset{i=1:J-1}{\prod}\mb{B}(i)$. This iterative IC
process can be expressed as a linear operation on $\tilde{\mb{y}}$
as {\small\be\label{eq-remaindb}
\tilde{\mb{y}}(J-1)=\mb{B}\tilde{\mb{y}}=\sqrt{\frac{P^2}{MP+M}}\mb{B}\mb{G}_1\left[\begin{array}{c}s^{(1)}_1\\\ol{s^{(1)}_2}\end{array}\right]+\underset{\mb{n}}{\underbrace{\sqrt{\frac{P}{MP+M}}\mb{B}\mb{G}_1\tilde{\mb{v}}^{(1)}+\mb{B}\tilde{\mb{w}}}}.
\ee} \hspace{-4pt}where {\small$\tilde{\mb{w}}=
\left[\tilde{\mb{w}}_1^*\ \cdots \ \tilde{\mb{w}}_N^*\right]^*$} and
$\mb{n}$ denotes the equivalent noise vector after IC. Note that
${\mb{n}}$ is Gaussian but not white. After straightforward
calculation, its covariance matrix can be calculated as
\be\label{eq-noisecov2}\mb{R}_\mb{n}=\frac{c_1^2}{\underset{i=1:M}{\sum}\left|f_i^{(1)}\right|^2}\mb{B}\mb{G}_1\mb{G}_1^*\mb{B}^*+\mb{B}\mb{B}^*,\ee
where $c_1=\sqrt{\frac{P}{MP+M}}$. The ML decoding of Source 1's
information can be performed as
{\small\begin{equation}\label{eq-MLTDMA} \arg
\min_{s_1^{(1)},s_2^{(1)} \in
\mc{S}}\left(\mb{B}\tilde{\mb{y}}-\sqrt{P}c_1\mb{B}\mb{G}_1\left[\begin{array}{c}s_1^{(1)}\\\ol{s_2^{(1)}}\end{array}\right]\right)^*\mb{R}_\mb{n}^{-1}\left(\mb{B}\tilde{\mb{y}}-\sqrt{P}c_1\mb{B}\mb{G}_1\left[\begin{array}{c}s_1^{(1)}\\\ol{s_2^{(1)}}\end{array}\right]\right).
\end{equation}}
Since the $2\times 2$ submatrices of $\mb{B}$, $\mb{G}_1$, and
$\mb{R}_\mb{n}$ have Alamouti structure, \eqref{eq-MLTDMA} can be
further decoupled into two procedures of symbol-wise decoding
following the similar argument in Subsection \ref{subsec-motivate}.
Likewise, the information of the other $J-1$ sources can be
decoupled and decoded. In total, $2J$ symbol-wise ML decoding
procedures are required to decode all sources' information.
Therefore, the decoding complexity is linear in the number of
sources.

\subsubsection{\TDMA-ICRec for the $1_J\times (4J)_1\times N_1$ MARN}
In this subsection, we describe \TDMA-ICRec for the MARN where the
number of relay antennas is four times the number of source
antennas. During the first step, Source $j$ collects four symbols
$s_i^{(j)}\ (i=1 ,\ldots ,4)$, in which $s_1^{(j)}, s_2^{(j)} \in
\mc{S}$ and $s_3^{(j)}, s_4^{(j)} \in \mc{S}'$. The constellation
$\mc{S}'$ is obtained by rotating $\mc{S}$ \cite{ShPa03, SuXia04}.
Source $j$ transmits a vector of these four symbols to the relay in
four time slots, and sources timeshare the source-relay link. The
relay obtains a soft estimate of each symbol from Source $j$ by
coherently combining signals at different antennas as in
\eqref{eq-relayMRCdb}, then linearly transforms this soft estimate
$\hat{\br}^{(j)}$ into a quasi-orthogonal DSTC by $[\mb{t}_{4j-3}\
\mb{t}_{4j-2}\ \mb{t}_{4j-1}\
\mb{t}_{4j}]=c_2\left[\mb{A}_1\hat{\mb{r}}^{(j)}\
\mb{B}_2\ol{\hat{\mb{r}}^{(j)}} \ \mb{B}_3\ol{\hat{\mb{r}}^{(j)}}\
\mb{A}_4\hat{\mb{r}}^{(j)}\right]$, where the scalar
$c_2=\sqrt{\frac{P}{MP+M}}$ normalizes the average power at the
relay to $P$; and $\mb{A}_i$ and $\mb{B}_i$ are the DSTC encoding
matrices\cite{DSTC-OD}, as given in \eqref{eq-quasides}.

During the second step, information streams from all sources are
concurrently forwarded to the destination by sending $\mb{t}_i$ from
the $i$-th antenna. With this design, the relay uses four of its
antennas to forward the quasi-orthogonal DSTC of each source and the
information of all sources is forwarded concurrently. Denote
$y_{\tau n}$ and $w_{\tau n}$ as the sampled signal and noise at the
$n$-th destination antenna and time slot $\tau$, respectively.
Following the analysis in \cite{KaJa}, two equivalent Alamouti
systems can be obtained as{\small\setlength{\arraycolsep}{1pt}
\begin{eqnarray*}
\underset{\mb{y}_n^+}{\underbrace{\left[\begin{array}{c}y_{1n}+y_{4n}\\
\ol{y_{2n}}-\ol{y_{3n}}\end{array}\right]}}
&=&\sqrt{P}c_2\sum_{j=1:J}\underset{\mb{G}_{jn}^+}{\underbrace{\left[\begin{array}{cc}g_{1n}^{(j)}+g_{4n}^{(j)}&g_{2n}^{(j)}-g_{3n}^{(j)}\\\ol{g_{2n}^{(j)}}-\ol{g_{3n}^{(j)}}&-\ol{g_{1n}^{(j)}}-\ol{g_{4n}^{(j)}}\end{array}\right]}}\underset{\mb{s}^{(j)+}}{\underbrace{\left[\begin{array}{c}s_1^{(j)}+s_4^{(j)}\\ \ol{s_3^{(j)}}-\ol{s_2^{(j)}}\end{array}\right]}}+c_2\sum_{j=1:J}\mb{G}_{jn}^+\underset{\mb{v}^{(j)+}}{\underbrace{\left[\begin{array}{c}\tilde{v}_1^{(j)}+\tilde{v}_4^{(j)}\\ \ol{\tilde{v}_3^{(j)}}-\ol{\tilde{v}_2^{(j)}}\end{array}\right]}}+\underset{\mb{w}_n^+}{\underbrace{\left[\begin{array}{c}w_{1n}+w_{4n}\\\ol{w_{2n}}-\ol{w_{3n}}\end{array}\right]}}\\
\underset{\mb{y}_n^-}{\underbrace{\left[\begin{array}{c}y_{1n}-y_{4n}\\\ol{y_{2n}}+\ol{y_{3n}}\end{array}\right]}}&=&\sqrt{P}c_2\sum_{j=1:J}\underset{\mb{G}_{jn}^-}{\underbrace{\left[\begin{array}{cc}g_{1n}^{(j)}-g_{4n}^{(j)}&g_{2n}^{(j)}+g_{3n}^{(j)}\\\ol{g_{2n}^{(j)}}+\ol{g_{3n}^{(j)}}&-\ol{g_{1n}^{(j)}}+\ol{g_{4n}^{(j)}}\end{array}\right]}}\underset{\mb{s}^{(j)-}}{\underbrace{\left[\begin{array}{c}s_1^{(j)}-s_4^{(j)}\\
-\ol{s_3^{(j)}}-\ol{s_2^{(j)}}\end{array}\right]}}+c_2\sum_{j=1:J}\mb{G}_{jn}^-\underset{\mb{v}^{(j)-}}{\underbrace{\left[\begin{array}{c}\tilde{v}_1^{(j)}-\tilde{v}_4^{(j)}\\
-\ol{\tilde{v}_3^{(j)}}-\ol{\tilde{v}_2^{(j)}}\end{array}\right]}}+\underset{\mb{w}_n^-}{\underbrace{\left[\begin{array}{c}w_{1n}-w_{4n}\\\ol{w_{3n}}+\ol{w_{2n}}\end{array}\right]}},
\end{eqnarray*}}
\hspace{-4pt}where $g_{kn}^{(j)}$ $(k=1, 2, 3, 4)$ denotes the four
channel paths from the four relay antennas that forward Source $j$
's DSTC to the $n$-th destination antenna, i.e.,
$g_{kn}^{(j)}=g_{(4j-4+k)n}$; $\tilde{v}_k^{(j)}$ denotes the
equivalent noises at the relay, which can be shown to be i.i.d.
{\small$\mc{CN}\left(0,\left(\underset{i=1:M}{\sum}\left|f_{i}^{(j)}\right|^2\right)^{-1}\right)$}.
By applying the multi-user IC proposed for quasi-orthogonal STBC in
\cite{KaJa}, the destination can cancel the symbols of Sources 2 to
J for each Alamouti system. Denote $\mb{B}^{\star}$ as the IC matrix
for system $\star$, $(\star=+,-)$, which can be obtained similarly
following the iterative IC process in Subsection
\ref{subsec-subsub}. Let $\mb{G}_1^{\star}=\left[\mb{G}_{11}^{\star
t}\ \cdots\ \mb{G}_{1N}^{\star t}\right]^t$ and
$\mb{y}^\star=\left[\mb{y}_1^{\star t}\ \cdots\ \mb{y}_N^{\star
t}\right]^t$. The resulting system equation for Source 1's
information after the IC can be expressed
as{\small\setlength{\arraycolsep}{1pt}
\be\label{eq-remainfour}\left[\begin{array}{c}\mb{B}^+\mb{y}^+\\\mb{B}^-\mb{y}^-\end{array}\right]=\sqrt{P}c_2\underset{\mb{H}}{\underbrace{\left[\begin{array}{cc}\mb{B}^+\mb{G}_1^{+}&\mb{0}\\\mb{0}&\mb{B}^-\mb{G}_1^{-}\end{array}\right]}}\left[\begin{array}{c}s_1^{(1)}+s_4^{(1)}\\\ol{s_3^{(1)}}-\ol{s_2^{(1)}}\\s_1^{(1)}-s_4^{(1)}\\-\ol{s_3^{(1)}}-\ol{s_2^{(1)}}\end{array}\right]+\underset{\mb{n}}{\underbrace{c_2\left[\begin{array}{cc}\mb{B}^+\mb{G}_1^{+}&\mb{0}\\\mb{0}&\mb{B}^-\mb{G}_1^{-}\end{array}\right]\left[\begin{array}{c}\mb{v}^{(1)+}\\
\mb{v}^{(1)-}\end{array}\right]+\left[\begin{array}{cc}\mb{B}^{+}&\mb{0}\\\mb{0}&\mb{B}^{-}\end{array}\right]\mb{w}}},
\ee} \hspace{-4pt}where $\mb{w}=\left[\mb{w}_1^{+*}\ \cdots\
\mb{w}_N^{+*}\ \mb{w}_1^{-*} \ \cdots \ \mb{w}_N^{-*}\right]^*$.
$\mb{H}$ and $\mb{n}$ denote the $4(N-J+1)\times 4$ equivalent
channel matrix and the $4(N-J+1)\times 1$ equivalent noise vector,
respectively. From \eqref{eq-remainfour}, it can be shown that two
procedures of pair-wise ML decoding are sufficient to decode the
four symbols of Source 1.

\subsubsection{\TDMA-ICRec for general $1_J\times M_1\times N_1$ MARNs}
For general $1_J\times M_1\times N_1$ MARNs, each source transmits a
vector of $2^n$ symbols in TDMA during the first step, where $2^n$
is the minimum number that is no less than
$\lfloor\frac{M}{J}\rfloor$. The relay constructs one $2^n\times
2^n$ DSTC using the quasi-orthogonal STBCs with ABBA structure for
each source\cite{QOD,TBH00}. During the second step, the first
$\lfloor\frac{M}{J}\rfloor$ columns of each DSTC are forwarded using
$\lfloor\frac{M}{J}\rfloor$ antennas of the relay. All DSTCs are
concurrently forwarded to the destination. The destination separates
the equivalent system into $2^{n-1}$ Alamouti
systems\cite{ICRelay-TDMA-jou}, then decouples information of each
source by IC \cite{KaJa}, after which, decodes each source's
information independently.

\subsection{Diversity Gain Analysis}\label{subsec-divana2}
In this subsection, we analyze the achievable diversity gain of
\TDMA-ICRec. As discussed in Subsection \ref{subsec-divupperbound},
the diversity gain can be calculated using the outage probability of
the instantaneous normalized receive SNR as in \eqref{eq-div}. To
help the presentation, we use an equivalent representation of
\eqref{eq-div}. We say that an instantaneous normalized receive SNR
$\gamma$ provides a diversity gain of $d$ if
$P(\gamma<\epsilon)=\alpha_1\epsilon^d+o(\epsilon^d)$ with
$\alpha_1$ independent of $\epsilon$. To calculate the diversity
gain of \TDMA-ICRec, the following lemma is
used\cite{ICRelay-TDMA-jou}.
\begin{lemma}
Let $\gamma_1, \gamma_2, \ldots, \gamma_k, \gamma_g$ be $k+1$
instantaneous normalized receive SNRs. $\gamma_g$ is independent of
$\gamma_n$ for $n=1, 2,\ldots, k$. $\gamma_g$ provides a diversity
gain of $d_1$; $\underset{n=1:k}{\sum}\gamma_n$ provides a diversity
gain of $d_2$. If
$\gamma=\underset{n=1:k}{\sum}\frac{\gamma_n\gamma_g}{\gamma_n+\gamma_g}$,
$\gamma$ provides a diversity gain of $\min\{d_1, d_2\}$.
\end{lemma}

Here is the theorem on the diversity gain of \TDMA-ICRec.

\begin{theorem}\label{thm-TDMAICRec2}
In $1_J\times M_1\times N_1$ MARNs, \TDMA-ICRec achieves a diversity
gain of $\min\{M, \lfloor\frac{M}{J}\rfloor(N-J+1)\}$.
\end{theorem}
\begin{proof}
See Appendix \ref{ap-thm2}.
\end{proof}

Intuitively, the result in Theorem \ref{thm-TDMAICRec2} can be
explained as follows. From the protocol design, since sources are
assigned to orthogonal channels in the first step of transmission,
the maximum diversity gain that can be achieved in the source-relay
link is $M$. For the second step, each source is allocated
$\lfloor\frac{M}{J}\rfloor$ antennas of the relay. Then, the
transmit diversity gain is $\lfloor\frac{M}{J}\rfloor$. Similar to
MAC systems, the destination uses $J-1$ antennas to cancel the
symbols of interfering sources and obtains the full receive
diversity gain $N-J+1$ at remaining antennas. Thus, the maximum
achievable diversity gain in the relay-destination link is
$\lfloor\frac{M}{J}\rfloor(N-J+1)$. The overall diversity gain of
\TDMA-ICRec is thus upperbounded by the minimum of the two values,
and Theorem \ref{thm-TDMAICRec2} shows that \TDMA-ICRec achieves
this upperbound. When $M=2J$ and $M=4J$, the following corollary is
obtained for special cases of Theorem \ref{thm-TDMAICRec2}.

\begin{corollary}\label{cor-TDMAICRec2}
In the $1_J\times (2J)_1\times N_1$ MARN, \TDMA-ICRec achieves a
diversity gain of $2\min\{J, N-J+1\}$; In the $1_J\times
(4J)_1\times N_1$ MARN, \TDMA-ICRec achieves a diversity gain of
$4\min\{J, N-J+1\}$.
\end{corollary}

\subsection{Discussion}\label{subsec-comp}
In this subsection, we discuss several properties of \TDMA-ICRec,
including the constraint on the number of sources, the condition to
achieve the int-free diversity gain, and the symbol rate. The
comparison of \TDMA-ICRec with other linear schemes is also
presented. Finally, we provide an application in the interference
relay network.

First, we discuss the constraint on the number of sources for
\TDMA-ICRec. Since in this protocol, each source is allocated a
different set of relay antennas for the concurrent transmission in
the relay-destination link, the number of relay antennas needs to be
no less than the number of sources, i.e., $M\ge J$. At the
destination, to fully decouple signals of different sources, at
least $J-1$ antennas are required to cancel the symbols of $J-1$
sources. In other words, $N\ge J$. Therefore, $J \le \min\{M, N\}$
is required. This condition is the same as that for \DSTC-ICRec as
discussed in Subsection \ref{subsec-dis}. To guarantee this
condition, user admission control in the upper layer is needed.

Next, we show that \TDMA-ICRec has the potential to achieve the
int-free diversity gain. Recall that the int-free diversity gain is
defined as the maximum achievable diversity gain when there is no
interference in both links. For $1_J\times M_1\times N_1$ MARNs, the
int-free diversity gain is $\min\{M, MN\}=M$, achievable by
\full-TDMA-DSTC\cite{DSTC-mulpaper}. Theorem \ref{thm-TDMAICRec2}
indicates that when \be\label{eq-intcond} N\ge
\frac{M}{\lfloor\frac{M}{J}\rfloor}+J-1,\ee \TDMA-ICRec achieves the
int-free diversity gain of $M$. Eq.~\eqref{eq-intcond} is called
\emph{the int-free condition}. For networks satisfying the int-free
condition, \TDMA-ICRec allows multi-source transmission in the
relay-destination link without sacrificing the diversity gain. If
$M$ is a multiple of $J$, this condition can be further simplified
as $N \ge 2J-1$. Examples of networks satisfying the int-free
condition are: $1_2\times 2_1\times 3_1$, $1_2\times 4_1\times 3_1$,
$1_3\times 3_1\times 5_1$, and $1_3\times 6_1\times 5_1$ MARNs.

In what follows, we discuss the symbol rate of \TDMA-ICRec. In
$1_J\times M_1\times N_1$ MARNs, for each source to transmit a
vector of $2^n$ symbols~($2^n$ is the minimum number no less than
$\lfloor\frac{M}{J}\rfloor$), $2^n J$ channel uses are needed in the
first link and $2^n$ channel uses are needed in the second link.
Thus, $2^n$ symbols from each source are transmitted using
$2^n(J+1)$ channel uses from end to end. The symbol rate can thus be
calculated as $R=\frac{2^n}{2^n(1+J)}=\frac{1}{1+J}$
symbols/source/channel use.

\TDMA-ICRec fits the linear framework in which the relay linearly
transforms its received signals to generate output signals without
decoding and the destination decouples signals from different
sources to separately decode each source's information. In what
follows, we compare the diversity gain, symbol rate, and CSI
requirements at the relay of \TDMA-ICRec with the two existing
linear schemes, \full-TDMA-DSTC and \IC-Relay-TDMA, as well as
\DSTC-ICRec. The results are shown in Table \ref{table-comp}. For
MARNs satisfying the int-free condition, \TDMA-ICRec and
\full-TDMA-DSTC achieve the maximum int-free diversity gain, higher
than that of \IC-Relay-TDMA. Both \TDMA-ICRec and \IC-Relay-TDMA
achieve higher symbol rate compared to \full-TDMA-DSTC. Thus,
\TDMA-ICRec outperforms \full-TDMA-DSTC in terms of symbol rate and
exceeds \IC-Relay-TDMA in terms of diversity gain. For \TDMA-ICRec
and \IC-Relay-TDMA, the relay needs to know its channels with all
sources, which can be obtained by training. For MARNs not satisfying
the int-free condition, \TDMA-ICRec may not achieve the int-free
diversity. For the two proposed protocols, each has its advantage
over the other: \TDMA-ICRec achieves a higher diversity gain,
whereas \DSTC-ICRec has a higher symbol rate.

\TDMA-ICRec can also be applied to the interference relay network
with one multi-antenna relay and several multi-antenna destinations.
The relay processes its received signals in the same way as that in
the MARN. Each destination cancels the information of undesired
sources and decodes the information of its interest as long as the
numbers of antennas at each destination and the relay are no less
than the number of sources.

\section{Numerical Results}\label{sec-Simulation}
In this section, we present simulated BERs of \DSTC-ICRec and
\TDMA-ICRec and compare with the BERs of other existing schemes with
similar complexities and CSI requirements. Since the average power
constraints at all nodes are equal to $P$ and noises are normalized,
the average transmit SNR at each node is $P$. For all figures, the
horizontal axis represents the average transmit SNR, measured in dB;
the vertical axis represents the BER.

In Fig.~\ref{fig-DSTCICRec-BER}, the BER of the first proposed
scheme, \DSTC-ICRec, is demonstrated for 6 MARNs: $1_2\times
2_1\times 2_1$, $1_2\times 2_1\times 3_1$, $1_2\times 2_1\times
4_1$, $1_2\times 4_1\times 2_1$, $1_2\times 4_1\times 3_1$,
$1_2\times 4_1\times 4_1$, and $1_3\times 4_1\times 3_1$. BPSK
modulation is used. Fig.~\ref{fig-DSTCICRec-BER} shows that the
scheme achieves a diversity gain of 1 in the $1_2\times 2_1\times
2_1$, $1_2\times 2_1\times 3_1$, and $1_2\times 2_1\times 4_1$
MARNs. Additional array gain can be achieved when the number of
destination antennas is increased. In the $1_3\times 4_1\times 3_1$
MARN, the diversity gain is 2; while in the $1_2\times 4_1\times
2_1$, $1_2\times 4_1\times 3_1$, and $1_2\times 4_1\times 4_1$
MARNs, the diversity gain is slightly less than 3. This is because
of the $\log P$ factor in the error rate formula\cite{DSTC-OD}. As $P$ increases, the diversity gain approaches 3. 
These results justify the validity of the diversity upperbound
presented in Theorem \ref{thm-DSTCICRec1} and show the achievability
of the upperbound for these network scenarios. Comparing the results
for the $1_3\times 4_1 \times 3_1$ and $1_2\times 2_1 \times 3_1$
MARNs, we can see that the number of sources that a MARN can
accommodate and the diversity gain of a MARN can be improved
simultaneously by increasing the number of relay antennas. From the
results for the $1_2\times 4_1\times 3_1$ and $1_3\times 4_1\times
3_1$ MARNs, we conclude that with a fixed number of relay antennas,
the diversity gain decreases as the number of sources in the network
increases.

Fig.~\ref{fig-TDMAICRec-BER} exhibits the BER of the second proposed
scheme \TDMA-ICRec in 8 MARNs: $1_2\times 2_1\times 2_1$, $1_2\times
2_1\times 3_1$, $1_2\times 2_1\times 4_1$, $1_3\times 3_1\times
3_1$, $1_3\times 3_1\times 5_1$,
$1_2\times 4_1\times 2_1$, $1_2\times 4_1\times 3_1$, and $1_2\times 8_1\times 2_1$. 
In all scenarios, BPSK modulation is used. For MARNs with parameters
$1_2\times 2_1\times 2_1$, $1_3\times 3_1\times 3_1$, $1_2\times
4_1\times 2_1$, and $1_2\times 8_1\times 2_1$, \TDMA-ICRec achieves
the diversity gains of 1, 1, 2, and 4, respectively. Note that the
int-free diversity gains in these networks are 2, 3, 4, and 8,
respectively. \TDMA-ICRec does not achieve the int-free diversity
gain for these scenarios. For MARNs with parameters $1_2\times
2_1\times 3_1$, $1_2\times 2_1\times 4_1$, $1_3\times 3_1\times
5_1$, and $1_2\times 4_1\times 3_1$, \TDMA-ICRec achieves the
diversity gains of 2, 2, 3, and 4, respectively, which are the
int-free diversity gains. The parameters of these four networks
satisfy the int-free condition given in Eq.~\eqref{eq-intcond}. The
simulation results for these eight networks justify our diversity
result in Theorem \ref{thm-TDMAICRec2}.

In the following, we compare the proposed \DSTC-ICRec (Scheme 1) and
\TDMA-ICRec (Scheme 2) with other schemes: \IC-Relay-TDMA (Scheme
3), \full-TDMA-DSTC (Scheme 4), \TDMADFICRec (Scheme 5), and \joint
(Scheme 6). Schemes 3 and 4 are introduced in Section 1. To compare
our methods with schemes having decoding at the relay, Scheme 5 is
introduced. It is similar to Scheme 2 except that the relay conducts
the ML decoding based on the soft estimate in \eqref{eq-relayMRCdb}.
After that, symbols are re-encoded and forwarded to the destination
using the same constellation. Scheme 6 is similar to Scheme 1, but
in Scheme 6 the destination jointly decodes all sources' information
without IC. Note that Schemes 1,2,3,4 satisfy the constraints of the
linear framework, but Schemes 5 and 6 do not fit the linear
framework and have higher complexity than the other four schemes.
For fair comparison in the numerical experiments, we fix the bit
rate to be 1 bit/source/channel use regardless of the scheme and
plot the BERs of the schemes as a function of the average transmit
SNR. Thus, QPSK, 8PSK, 8PSK, 16PSK, 8PSK, and QPSK are used for
Schemes 1, 2, 3, 4, 5, and 6, respectively.

Figs.~\ref{fig-comp1}, \ref{fig-comp2}, and \ref{fig-comp3} show
BERs of these schemes in the $1_2\times 2_1\times 2_1$, $1_2\times
2_1\times 3_1$, $1_2\times 4_1\times 3_1$ MARNs, respectively. We
compare the BERs of the four linear schemes. We first look at the
$1_2\times 2_1\times 2_1$ MARN whose BERs are shown in
Fig.~\ref{fig-comp1}. Only Scheme 4 achieves the maximum int-free
diversity, thus it has the best performance at high SNR (26 dB and
up). The other three schemes have a diversity gain of 1. Scheme 3
has the lowest BER for SNR less than 26 dB, because of its high
signal to interference-plus-noise ratio (SINR) at the destination.
Thus, for the $1_2\times 2_1\times 2_1$ MARN, the proposed schemes,
Schemes 1 and 2, are inferior in BER. The next is the $1_2\times
2_1\times 3_1$ MARN. We can see from Fig.~\ref{fig-comp2} that only
Schemes 2 and 4 achieve the maximum int-free diversity, while Scheme
2 has lower BER than the other schemes for all the simulated SNR
values. Its advantage over Scheme 4 is about 5~dB. This is because
Scheme 2 has a higher symbol rate. For the same bit rate, it can use
a smaller constellation, which provides higher array gain. Thus, for
the $1_2\times 2_1\times 3_1$, Scheme 2 is the best. Finally, in the
$1_2\times 4_1\times 3_1$ MARN shown in Fig.~\ref{fig-comp3}, Scheme
2 has the highest diversity gain. When SNR is higher than 23~dB,
Scheme 2 has the lowest BER. Scheme 1 outperforms the other three in
the SNR regime from 17 to 23~dB, because it has the highest symbol
rate and uses the smallest constellation to achieve the same bit
rate. When the SNR is smaller than 17~dB, Scheme 3 has the lowest
BER. Therefore, for the $1_2\times 4_1\times 3_1$ MARN, our proposed
two schemes have lower BER compared to the existing schemes when SNR
is higher than 17~dB. We can conclude from the three experiments
that the relative quality of the four schemes depends on the network
parameters and SNR range. The proposed Scheme 1, \DSTC-ICRec, is
expected to have good reliability in the low to moderate SNR range,
as observed in Fig.~\ref{fig-comp3}. The proposed Scheme 2,
\TDMA-ICRec, is expected to have good reliability for MARNs whose
relay-destination link is much stronger than the source-relay link
(e.g., the $1_2\times 2_1\times 3_1$ MARN). These are due to the
nature of the design explained in Section \ref{sec-DSTC-ICRec} and
Section \ref{sec-TDMA-ICRec}.

In what follows, we compare the proposed schemes with Schemes 5 and
6, which do not satisfy the linear constraints. We first compare
Scheme 1 with Scheme 6. Note that the ML decoding of Scheme 1 is
symbol-wise in the $1_2\times 2_1\times N_1$ MARN and pair-wise in
the $1_2\times 4_1\times N_1$ MARN, while for Scheme 6, the
destination needs to jointly decode four symbols in the $1_2\times
2_1\times N_1$ MARN and eight symbols in the $1_2\times 4_1\times
N_1$ MARN. For networks with large $J$ and $M$, the decoding
complexity of Scheme 6 is exponential in $JM$, thus becomes
impractical. For Scheme 1, the decoding complexity is linear in $J$
and exponential in $M/2$, thus is much lower. Figs.~\ref{fig-comp1},
\ref{fig-comp2}, and \ref{fig-comp3} show that this extra decoding
complexity can improve both diversity gain and array gain. The
diversity gain improvements are 1 in all three networks. For the
$1_2\times 4_1\times 3_1$ MARN~(Fig.~\ref{fig-comp3}), the array
gain improvement is the smallest~(about 3~dB at BER$=10^{-2}$)
compared to the other two networks. Therefore, compared to Scheme 6,
Scheme 1 is desired in large networks to trade performance
degradation for lower complexity. Then, we compare Scheme 2 with
Scheme 5 and see whether the extra decoding at the relay can provide
better performance. For the three networks shown in
Figs.~\ref{fig-comp1}, \ref{fig-comp2}, and \ref{fig-comp3}, Scheme
2 has the same diversity gain as Scheme 5. For the $1_2\times
2_1\times 2_1$ MARN (Fig.~\ref{fig-comp1}), Scheme 2 has
approximately the same performance as Scheme 5 for all the simulated
SNR values. For the $1_2\times 2_1\times 3_1$ MARN
(Fig.~\ref{fig-comp2}), Scheme 2 has the same performance as Scheme
5 in the high SNR regime while is about 1~dB worse in the low to
moderate SNRs. For the $1_2\times 4_1\times 3_1$ MARN
(Fig.~\ref{fig-comp3}), Scheme 2 is approximately 2~dB worse for all
SNRs. These observations can be explained as follows. For the
$1_2\times 2_1\times 2_1$ MARN, with Scheme 2, the BER of the
network is mainly constrained by the second hop (with IC at the
destination, the second hop has only a diversity gain of 1, but the
first hop has a diversity gain of 2). The extra decoding at the
relay only improves the performance in the first hop, does not help
the overall performance much. For the $1_2\times 2_1\times 3_1$ and
$1_2\times 4_1\times 3_1$ MARNs, with Scheme 2, the two links have
similar qualities (both links have diversity gain 2 for the
$1_2\times 2_1\times 3_1$ MARN and 4 for the $1_2\times 4_1\times
3_1$ MARN). The extra decoding complexity at the relay can provide a
better performance.

\section{Conclusions}\label{sec-Conclusion}
This paper studies multi-source transmission schemes for $1_J\times
M_1\times N_1$ MARNs. For complexity considerations, a linear
framework is introduced, where the relay conducts linear
transformation without decoding and the destination decouples
signals from different sources so that the decoding complexity is
linear in the number of sources. We propose two protocols that use
multi-antennas at the destination to resolve multi-source
interference. The protocol of \DSTC-ICRec allows concurrent
transmission of information streams from multi-sources in both the
source-relay link and the relay-destination link. The relay performs
DSTC and does not require any CSI. The destination uses the
multi-antenna IC technique to decouple signals from different
sources. \DSTC-ICRec achieves a symbol rate of $1/2$
symbols/source/channel use, but its diversity gain is shown to be
upperbounded by $M-J+1$. Thus, for this protocol, the diversity gain
degradation is necessary to trade for symbol rate. To improve the
diversity gain, we propose \TDMA-ICRec, in which concurrent
transmission is allowed in the relay-destination link but TDMA is
used in the source-relay link. After receiving signals from the
sources, the relay first conducts MRC to maximize the SNR of each
source then concurrently transmits all sources' information to the
destination using DSTC. At the destination, IC is performed to
decouple signals from different sources before decoding. Through
analysis and simulations, it is shown that \TDMA-ICRec achieves a
diversity gain of $\min\left\{M,
\lfloor\frac{M}{J}\rfloor(N-J+1)\right\}$ with a symbol rate of
$\frac{1}{J+1}$. When $N \ge 2J-1$, \TDMA-ICRec achieves the same
maximum int-free diversity gain of the network but with a higher
symbol rate, compared to a full TDMA scheme.

\section*{Appendix}
\appendix
\section{Proof of Theorem \ref{thm-DSTCICRec1}}\label{ap-thm1} To
prove this theorem, it suffices to find an upperbound on the
instantaneous normalized receive SNR. We first show the scenario of
$J=2$ and $M=2$, then its generalization.

From \eqref{eq-DSTCICRec-ML}, the noise covariance matrix
$\mb{R}_\mb{n}$ can be lowerbounded by $\mb{R}_\mb{n}\succ
\mb{B}\mb{B}^*$. Then, the instantaneous normalized receive SNR for
$s_1^{(1)}$ can be upperbounded as \be\label{eq-SNRup}
\gamma=(\mb{B}\tilde{\mb{G}}\mb{\Phi}
\hat{\mb{f}}_1^{(1)})^*\mb{R}_\mb{n}^{-1}(\mb{B}\tilde{\mb{G}}\mb{\Phi}
\hat{\mb{f}}_1^{(1)})<\hat{\mb{f}}_1^{(1)*}\mb{\Phi}^*\tilde{\mb{G}}^*\mb{B}^*(\mb{B}\mb{B}^*)^{-1}\mb{B}\tilde{\mb{G}}\mb{\Phi}
\hat{\mb{f}}_1^{(1)}<\hat{\mb{f}}_1^{(1)*}\mb{\Phi}^*\tilde{\mb{G}}^*\tilde{\mb{G}}\mb{\Phi}
\hat{\mb{f}}_1^{(1)}, \ee  where $\hat{\mb{f}}^{(j)}_i$ denotes the
$i$-th column of $\mb{F}^{(j)}$. For the second inequality we have
used the fact that $\mb{B}^*(\mb{BB}^*)^{-1}\mb{B}\prec
\mb{I}_{2N}$. Since $\hat{\mb{f}}_1^{(1)}$ is orthogonal to
$\hat{\mb{f}}_2^{(2)}$ from \eqref{eq-remaining2}, the projection
$\mb{\Phi}\hat{\mb{f}}_1^{(1)}$ is equivalent to project
$\hat{\mb{f}}_1^{(1)}$ onto the null space of $\hat{\mb{f}}_1^{(2)}$
only , i.e.,
$\mb{\Phi}\hat{\mb{f}}_1^{(1)}=\mb{\Xi}\hat{\mb{f}}_1^{(1)},$ where
{\small$\mb{\Xi}=\left(\mb{I}_4-\frac{\hat{\mb{f}}_1^{(2)}\hat{\mb{f}}_1^{(2)*}}{\|\hat{\mb{f}}_1^{(2)}\|^2}\right)$}.
Note that $\tilde{\mb{G}}^*\tilde{\mb{G}}\prec
\tr\left(\tilde{\mb{G}}^*\tilde{\mb{G}}\right)\mb{I}_4$. The
right-hand side (RHS) of \eqref{eq-SNRup} can be further
upperbounded by \be\label{eq-divproof} \gamma<\tr
(\tilde{\mb{G}}^*\tilde{\mb{G}})
\hat{\mb{f}}_1^{(1)*}\mb{\Xi}^*\mb{\Xi}
\hat{\mb{f}}_1^{(1)}=2\underset{g}{\underbrace{\sum_{n=1:N}\left(|g_{1n}|^2+|g_{2n}|^2\right)}}\hat{\mb{f}}_1^{(1)*}\mb{\Xi}
\hat{\mb{f}}_1^{(1)}=2g\underset{f}{\underbrace{\mb{f}^{(1)*}\mb{\Theta}
\mb{f}^{(1)}}},\ee where the first equality holds because
$\mb{\Xi}^*\mb{\Xi}=\mb{\Xi}$ from the definition of projection;
$\mb{f}^{(j)}$ is a $2\times 1$ channel vector from Source $j$ to
the relay, i.e., $\mb{f}^{(j)}= \left[f_1^{(j)}\
f_2^{(j)}\right]^t$; and
{\small$\mb{\Theta}=\mb{I}_2-\frac{{\mb{f}}^{(2)}\mb{f}^{(2)*}}{\|\mb{f}^{(2)}\|^2}$},
a $2\times 2$ projection matrix to the null space of
${\mb{f}}^{(2)}$. Clearly, the random variable $g$ is Gamma
distributed with degree $2N$. Next, we show that given
$\mb{f}^{(2)}$, $f$ is Gamma distributed with degree 1. Note that
the two eigenvalues of $\mb{\Theta}$ are 1 and 0. Then,
${\mb{f}}^{(1)*}\mb{\Theta}
{\mb{f}}^{(1)}={\mb{f}}^{(1)*}\mb{u}_1\mb{u}_1^* {\mb{f}}^{(1)}$,
where $\mb{u}_1$ is the eigenvector corresponding to the eigenvalue
1. Since $\mb{\Theta}$ depends on $\mb{f}^{(2)}$ only, $\mb{u}_1$ is
independent of ${\mb{f}}^{(1)}$. Thus, given $\mb{f}^{(2)}$,
$\mb{u}_1^*{\mb{f}}^{(1)}$ is $\mc{CN}(0,1)$ distributed and $f$ is
Gamma distributed with degree 1. The outage probability of $\gamma$
can be bounded as {\setlength{\arraycolsep}{1pt}\begin{eqnarray*}
P(\gamma<\epsilon)&=&\underset{\mb{f}^{(2)},g}{\Exp}\left[P\left(\gamma<\epsilon|\mb{f}^{(2)},g\right)\right]>\underset{\mb{f}^{(2)},g}{\Exp}\left[P\left(2gf<\epsilon|\mb{f}^{(2)},g\right)\right]=\underset{\mb{f}^{(2)},g}{\Exp}\left[P\left(f<\frac{\epsilon}{2g}|\mb{f}^{(2)},g\right)\right]\\
&=&\underset{\mb{f}^{(2)},g}{\Exp}\left[\alpha\frac{\epsilon}{2g}\right]+o(\epsilon)=\alpha\underset{g}{\Exp}\left[\frac{\epsilon}{2g}\right]+o(\epsilon)=\frac{\alpha}{2(2N-1)}\epsilon+o(\epsilon).
\end{eqnarray*}}
\hspace{-2pt}where $\alpha$ is a constant independent of $\epsilon$.
By \eqref{eq-div}, the diversity gain is upperbounded by one.

For MARNs with general $J$ and $M$, the IC operation at the
destination similarly creates a virtual ZF operation at the relay as
discussed in Subsection \ref{subsec-equivalent}. The virtual ZF
matrix $\mb{\Phi}$ nulls out $\mb{F}^{(j)},\ j=2,\ldots, J$.
Following a similar process, it can be shown that the instantaneous
normalized receive SNR is upperbounded by a product of two parts as
in \eqref{eq-divproof}: the first part depends on $g_{in}$ only, the
second part is equal to $\mb{f}^{(1)*}\mb{\Theta}\mb{f}^{(1)}$ where
$\mb{f}^{(j)}$ is the $M\times 1$ channel vector from Source $j$ to
the relay and $\mb{\Theta}$ is a projection matrix onto the null
spaces of $\mb{f}^{(2)}$ to $\mb{f}^{(J)}$. Similarly, the diversity
gain can be shown to be no higher than $M-J+1$.

\section{Proof of Theorem \ref{thm-TDMAICRec2}}\label{ap-thm2}
We first show the case that $M=2J$, then $M=4J$, followed by the
general scenario. When $M=2J$, the channel vector experienced by
$s^{(1)}_1$ is $\mb{B}\mb{g}_1$ from \eqref{eq-remaindb}, where
$\mb{g}_1$ denotes the first column of $\mb{G}_1$. Note that the
noise covariance matrix is given in \eqref{eq-noisecov2}. By the
definition of instantaneous normalized receive SNR, we have
{\setlength{\arraycolsep}{2pt}\small\begin{align}&\gamma=\mb{g}_1^*\mb{B}^*\left(\frac{c_1^2}{\sum
\left|f_{i}^{(1)}\right|^2}\mb{B}\mb{G}_1\mb{G}_1^*\mb{B}^*+\mb{B}\mb{B}^*\right)^{-1}\mb{B}\mb{g}_1\label{eq-line1}\\
\label{eq-line2}&=\mb{g}_1^*\mb{B}^*\left((\mb{BB}^*)^{-1}-(\mb{BB^*})^{-1}\mb{B}\mb{G}_1\left(\frac{\sum
\left|f_{i}^{(1)}\right|^2}{c_1^2}\mb{I}_2+\mb{G}_1^*\mb{B}^*(\mb{BB}^*)^{-1}\mb{B}\mb{G}_1\right)^{-1}\mb{G}_1^*\mb{B}^*(\mb{BB}^*)^{-1}\right)\mb{B}\mb{g}_1\\
\label{eq-line3}&=\mb{g}_1^*\mb{B}^*(\mb{BB}^*)^{-1}\mb{B}\mb{g}_1-\left(\frac{\sum
\left|f_{i}^{(1)}\right|^2}{c_1^2}+\mb{g}_1^*\mb{B}^*(\mb{BB}^*)^{-1}\mb{B}\mb{g}_1\right)^{-1}\mb{g}_1^*\mb{B}^*(\mb{BB}^*)^{-1}\mb{B}\mb{G}_1\mb{G}_1^*\mb{B}^*(\mb{BB}^*)^{-1}\mb{B}\mb{g}_1\\
\label{eq-line4}&=\mb{g}_1^*\mb{B}^*(\mb{BB}^*)^{-1}\mb{B}\mb{g}_1-\left(\frac{\sum
\left|f_{i}^{(1)}\right|^2}{c_1^2}+\mb{g}_1^*\mb{B}^*(\mb{BB}^*)^{-1}\mb{B}\mb{g}_1\right)^{-1}\left(\mb{g}_1^*\mb{B}^*(\mb{BB}^*)^{-1}\mb{B}\mb{g}_1\right)^2.
\end{align}}
\hspace{-2pt}From \eqref{eq-line1} to \eqref{eq-line2}, the matrix
inversion lemma is applied. For \eqref{eq-line3}, we use the fact
that {$\mb{G}_1^*\mb{B}^*(\mb{BB}^*)^{-1}\mb{B}\mb{G}_1$} is a
Hermitian matrix with Alamouti structure. Thus,
$\mb{G}_1^*\mb{B}^*(\mb{BB}^*)^{-1}\mb{B}\mb{G}_1$ is a $2\times 2$
diagonal matrix whose diagonal entries are equal to
$\mb{g}_1^*\mb{B}^*(\mb{BB}^*)^{-1}\mb{B}\mb{g}_1$. Eq.
\eqref{eq-line4} follows from \eqref{eq-line3} because the second
entry of the vector
$\mb{g}_1^*\mb{B}^*(\mb{BB}^*)^{-1}\mb{B}\mb{G}_1$ is zero. Let
{\small$y=\mb{g}_1^*\mb{B}^*(\mb{BB}^*)^{-1}\mb{B}\mb{g}_1$} and
{\small$x=\underset{i=1:M}{\sum} \left|f_{i}^{(1)}\right|^2$}. We
have \be\label{eq-SNRremain}\gamma=\frac{xy}{x+yc_1^2},\ee which is
a scaled harmonic mean of variables $x$ and $yc_1^2$. Since $x$ is
the sum of $M$ independent random variables with exponential
distribution, $x$ is Gamma distributed with degree $M$. Thus, $x$
has a diversity gain of $M$. For $yc_1^2$, when $P\gg 1$,
$c_1=\underset{P\rightarrow \infty}{\lim}
\sqrt{\frac{P}{MP+M}}\approx \frac{1}{\sqrt{M}}$. From the iterative
algorithm and \eqref{eq-ICmatrix2}, $\mb{B}$ depends on
$\mb{G}_{jn}$ for $j=2, \ldots, J$. On the other hand, $\mb{g}_1$
only depends on $\mb{G}_{1n}$. Thus,
$\mb{B}^*(\mb{BB}^*)^{-1}\mb{B}$ and $\mb{g}_1$ are independent.
Because $\mb{B}$ zero-forces $\mb{G}_2$ to $\mb{G}_J$, it can be
shown that $\mb{B}^*(\mb{BB}^*)^{-1}\mb{B}$ is a projection matrix
onto the null space of the subspace spanned by columns of $\mb{G}_j$
for $j=2,\ldots, J$. Then, $y$ is Gamma distributed with degree
$2(N-J+1)$, implying $c_1^2y$ has diversity gain $2(N-J+1)$. Let
$k=1$ in Lemma 1. The diversity gain of $\gamma$ is the smaller of
the diversities of $x$ and $y$. Then, the achievable diversity gain
of \TDMA-ICRec is $\min\{M, 2(N-J+1)\}=2\min\{J, N-J+1\}$.

For $M=4J$, the instantaneous normalized receive SNR can be shown to
be
$\gamma=\tr\left(\mb{H}^*\mb{R}_\mb{N}^{-1}\mb{H}\right)$\cite{ICRelay-TDMA-jou},
where $\mb{H}$ denotes the equivalent channel matrix and
$\mb{R}_\mb{N}$ denotes the covariance matrix of $\mb{n}$ in
\eqref{eq-remainfour}. After straightforward calculation, we have
{\small\[\mb{R}_\mb{N}=\diag
\left(\frac{2c_2^2}{\sum\left|f^{(1)}_{i}\right|^2}\mb{B}^+\mb{G}^{(1)+}\mb{G}^{(1)+*}\mb{B}^{+*}+\mb{B}^+\mb{B}^{+*},
\frac{2c_2^2}{\sum\left|f^{(1)}_{i}\right|^2}\mb{B}^-\mb{G}^{(1)-}\mb{G}^{(1)-*}\mb{B}^{-*}+\mb{B}^-\mb{B}^{-*}
\right).\]} By similar calculation procedures from \eqref{eq-line1}
to \eqref{eq-line4}, it follows that
\be\label{eq-SNRremainfour}\gamma=\frac{xy}{x+2c_2^2y}+\frac{xz}{x+2c_2^2z},\ee
where $x=\underset{i=1:M}{\sum}\left|f^{(1)}_{i}\right|^2$,
$y=\mb{g}_1^{+*}\mb{B}^{+*}(\mb{B}^+\mb{B}^{+*})^{-1}\mb{B}^+\mb{g}_1^+$,
and
$z=\mb{g}_1^{-*}\mb{B}^{-*}(\mb{B}^-\mb{B}^{-*})^{-1}\mb{B}^-\mb{g}_1^-$,
with $\mb{g}_1^{\star}$ the first column of $\mb{G}_1^{\star}$ for
$\star=+,-$. The random variable $x$ is Gamma distributed with
degree $M$ and has a diversity gain of $M$. The random variables $y$
and $z$ are independent and both have a diversity gain of
$2(N-J+1)$. Then, $y+z$ has a diversity gain of $4(N-J+1)$. Let
$k=2$ in Lemma 1. The achievable diversity gain of \TDMA-ICRec is
$\min\{M, 4(N-J+1)\}=4\min\{J, N-J+1\}$.

For MARNs with a general $M$, the relay encodes the information of
one source using one quasi-orthogonal DSTC with ABBA
structure\cite{QOD,TBH00} and forwards each codeword by
$\lfloor\frac{M}{J}\rfloor$ of its antennas. The destination
conducts the multi-user IC technique\cite{KaJa}. The proof for this
general case is a straightforward extension of the proofs for the
cases that $M=2J$ and $M=4J$. Thus, the diversity result in Theorem
\ref{thm-TDMAICRec2} follows.

{\footnotesize \renewcommand{\baselinestretch}{1}
\bibliographystyle{ieeetran}
\bibliography{IEEEabrv,IC-Relay-TDMA}}

\begin{thebibliography}{10}
\providecommand{\url}[1]{#1}
\csname url@samestyle\endcsname
\providecommand{\newblock}{\relax}
\providecommand{\bibinfo}[2]{#2}
\providecommand{\BIBentrySTDinterwordspacing}{\spaceskip=0pt\relax}
\providecommand{\BIBentryALTinterwordstretchfactor}{4}
\providecommand{\BIBentryALTinterwordspacing}{\spaceskip=\fontdimen2\font plus
\BIBentryALTinterwordstretchfactor\fontdimen3\font minus
  \fontdimen4\font\relax}
\providecommand{\BIBforeignlanguage}[2]{{%
\expandafter\ifx\csname l@#1\endcsname\relax
\typeout{** WARNING: IEEEtran.bst: No hyphenation pattern has been}%
\typeout{** loaded for the language `#1'. Using the pattern for}%
\typeout{** the default language instead.}%
\else
\language=\csname l@#1\endcsname
\fi
#2}}
\providecommand{\BIBdecl}{\relax}
\BIBdecl

\bibitem{LenemanWornell}
J.~Laneman and G.~Wornell, ``Distributed space-time-coded protocols for
  exploiting cooperative diversity in wireless network,'' \emph{IEEE Tran.~on
  Info.~Theory}, vol.~49, pp. 2415--2425, Oct. 2003.

\bibitem{DSTC-paper}
Y.~Jing and B.~Hassibi, ``Distributed space-time coding in wireless relay
  networks,'' \emph{IEEE Trans.~on Wireless Comm.}, vol.~5, pp. 3524--3536,
  Dec. 2006.

\bibitem{zz-eg-sc}
K.~Azarian, H.~Gamal, and P.~Schniter, ``On the achievable
  diversity-multiplexing tradeoff in half-duplex cooperative channels,''
  \emph{IEEE Trans.~on Info.~Theory}, vol.~51, pp. 4152--4172, Dec. 2005.

\bibitem{LaTsWo}
J.~Laneman, D.~Tse, and G.~Wornell, ``Cooperative diversity in wireless
  networks: Efficient protocols and outage behavior,'' \emph{IEEE Trans.~on
  Info.~Theory}, vol.~50, pp. 3062--3080, Dec. 2004.

\bibitem{L06}
L.~Venturino, X.~Wang, and M.~Lops, ``Multiuser detection for cooperative
  networks and performance analysis,'' \emph{IEEE Trans.~on Signal Processing},
  vol.~54, no.~9, pp. 3315 --3329, Sep. 2006.

\bibitem{O08}
O.~Oteri and A.~Paulraj, ``Multicell optimization for diversity and
  interference mitigation,'' \emph{IEEE Trans.~on Signal Processing}, vol.~56,
  no.~5, pp. 2050 --2061, May 2008.

\bibitem{K09}
K.~Zarifi, S.~Affes, and A.~Ghrayeb, ``Large-system-based performance analysis
  and design of multiuser cooperative networks,'' \emph{IEEE Trans.~on Signal
  Processing}, vol.~57, no.~4, pp. 1511 --1525, Apr. 2009.

\bibitem{Yilmaz-ICC}
A.~O. Yilmaz, ``Cooperative multiple-access in fading relay channels,'' in
  \emph{Proc.~of IEEE ICC}, Istanbul, Turkey, Jun. 2006.

\bibitem{MoBoNa05}
V.~Morgenshtern, H.~Bolcskei, and R.~Nabar, ``Distributed orthogonalization in
  large interference relay networks,'' in \emph{Proc.~of International
  Symposium on Information Theory, 2005.}, Adelaide, Australia, Sep. 2005, pp.
  1211 --1215.

\bibitem{WitRan04}
A.~Wittneben and B.~Rankov, ``Distributed antenna systems and linear relaying
  for gigabit {MIMO} wireless,'' in \emph{Proc.~of IEEE Vehicular Technology
  Conference VTC}, Los Angeles, USA, Fall 2004.

\bibitem{Wit06}
A.~Wittneben, ``Coherent multiuser relaying with partial relay cooperation,''
  in \emph{Proc.~of IEEE WCNC}, Las Vegas, NV, USA, Apr. 2006.

\bibitem{Niu07}
B.~Niu, O.~Simeone, O.~Somekh, and A.~M. Haimovich, ``Throughput of two-hop
  wireless networks with relay cooperation,'' in \emph{Proc.~of Allerton
  Conference}, Monticello, IL, Sep. 2007.

\bibitem{BerWit05}
S.~Berger and A.~Wittneben, ``Cooperative distributed multiuser {MMSE} relaying
  in wireless {A}d-{H}oc networks,'' in \emph{Asilomar Conference on Signals,
  Systems, and Computers 2005}, Pacific Grove, CA, Nov. 2005.

\bibitem{Keyi-ICCASP}
A.~El-Keyi and B.~Champagne, ``Cooperative {MIMO}-beamforming for multiuser
  relay networks,'' in \emph{Proc.~of IEEE ICCASP}, Las Vegas, NV, Apr. 2008.

\bibitem{Oyman07}
O.~Oyman and A.~Paulraj, ``Power-bandwidth tradeoff in dense multi-antenna
  relay networks,'' \emph{IEEE Trans.~on Wireless Comm.}, vol.~7, pp.
  2282--2292, Jun. 2007.

\bibitem{ICRelay-TDMA-jou}
L.~Li, Y.~Jing, and H.~Jafarkhani, ``Interference cancellation at the relay in
  multi-access wireless relay networks,'' \emph{submitted to IEEE Trans.~on
  Wireless Comm., also available on http://arxiv.org/abs/1004.3807}, Apr. 2010.

\bibitem{DSTC-OD}
Y.~Jing and H.~Jafarkhani, ``Using orthogonal and quasi-orthogonal designs in
  wireless relay networks,'' \emph{IEEE Trans.~on Info.~Theory}, pp.
  4106--4118, Nov. 2007.

\bibitem{DSTC-mulpaper}
Y.~Jing and B.~Hassibi, ``Diversity analysis of distributed space-time codes in
  relay networks with multiple transmit/receive antennas,'' \emph{EURASIP
  Jour.~on Advances in Signal Proc.}, vol. 2008, 2008, article ID 254573, 17
  pages, doi:10.1155/2008/254573.

\bibitem{NaSeCa}
A.~Naguib, N.~Seshadri, and A.~Calderbank, ``Applications of space-time block
  codes and interference suppression for high capacity and high data rate
  wireless systems,'' in \emph{Proc.~of Asilomar Conf.}, Pacific Grove, CA,
  Oct. 1998.

\bibitem{AlCa}
A.~Stamoulis, N.~Al-Dhahir, and A.~Calderbank, ``Further results on
  interference cancellation and space-time block codes,'' in \emph{Proc.~of
  Asilomar Conf.}, Pacific Grove, CA, Oct. 2001.

\bibitem{KaJa}
J.~Kazemitabar and H.~Jafarkhani, ``Multiuser interference cancellation and
  detection for users with more than two transmit antennas,'' \emph{IEEE
  Trans.~on Comm.}, pp. 574--583, Apr. 2008.

\bibitem{SunJing2010}
S.~Sun and Y.~Jing, ``Channel training and estimation in distributed space-time
  coded relay networks with multiple transmit/receive antennas,'' in
  \emph{Proc.~of IEEE WCNC}, Sydney, Australia, Apr. 2010.

\bibitem{KaJa-2}
J.~Kazemitabar and H.~Jafarkhani, ``Performance analysis of multiple-antenna
  multi-user detection,'' \emph{Information Theory and Applications Workshop},
  Jan. 2009.

\bibitem{hj}
H.~Jafarkhani, \emph{Space-Time Coding:~Theory and Practice}.\hskip 1em plus
  0.5em minus 0.4em\relax Cambridge University Press, 2005.

\bibitem{QOD}
------, ``A quasi-orthogonal space-time block codes,'' \emph{IEEE Transactions
  on Communications}, vol.~49, pp. 1-- 4, Jan. 2001.

\bibitem{TBH00}
O.~Tirkkonen, A.~Boariu, and A.~Hottinen, ``Minimal nonorthogonality rate 1
  space-time block code for 3+ {T}x antennas,'' in \emph{Proc. IEEE 6th Int.
  Symp. Spread-Spectrum Techniques and Applications (ISSSTA 2000)}, Parsippany,
  NJ, USA, Sep. 2000.

\bibitem{divthm_report}
L.~Li, Y.~Jing, and H.~Jafarkhani, ``Using instantaneous normalized receive
  {SNR} for diversity gain calculation,'' \emph{CPCC Technical Report,
  available at http://escholarship.org/uc/item/9511q6pf}, Sep. 2010.

\bibitem{ShPa03}
N.~Sharma and C.~Papadias, ``Improved quasi-orthogonal codes through
  constellation rotation,'' \emph{IEEE Transactions on Communications},
  vol.~51, no.~3, pp. 332 -- 335, Mar. 2003.

\bibitem{SuXia04}
W.~Su and X.-G. Xia, ``Signal constellations for quasi-orthogonal space-time
  block codes with full diversity,'' \emph{IEEE Transactions on Information
  Theory}, vol.~50, no.~10, pp. 2331 -- 2347, Oct. 2004.

\end{thebibliography}

\renewcommand{\baselinestretch}{1.4}
\begin{table}[!ht]
  \centering
  \caption{The diversity gain and symbol rate performance for linear schemes. (The schemes marked with * are proposed in this paper.)}\label{table-comp}
  \begin{tabular}{|c|c|c|c|}
    \hline
    Protocol & Diversity Gain & Symbol Rate & Relay Backward CSI\\
    \hline
    \DSTC-ICRec * & $\le M-J+1$ & $\frac{1}{2}$ & No\\
    \TDMA-ICRec * & $\min\{M, \lfloor\frac{M}{J}\rfloor(N-J+1)\}$ & $\frac{1}{J+1}$ & Yes\\
    \IC-Relay-TDMA & $M-J+1$ & $\frac{1}{J+1}$& Yes\\
    \full-TDMA-DSTC & $M$ & $\frac{1}{2J}$ & No \\
    \hline
  \end{tabular}
\end{table}

%
\begin{figure}[!ht]
\centering
\includegraphics[width=4.5in]{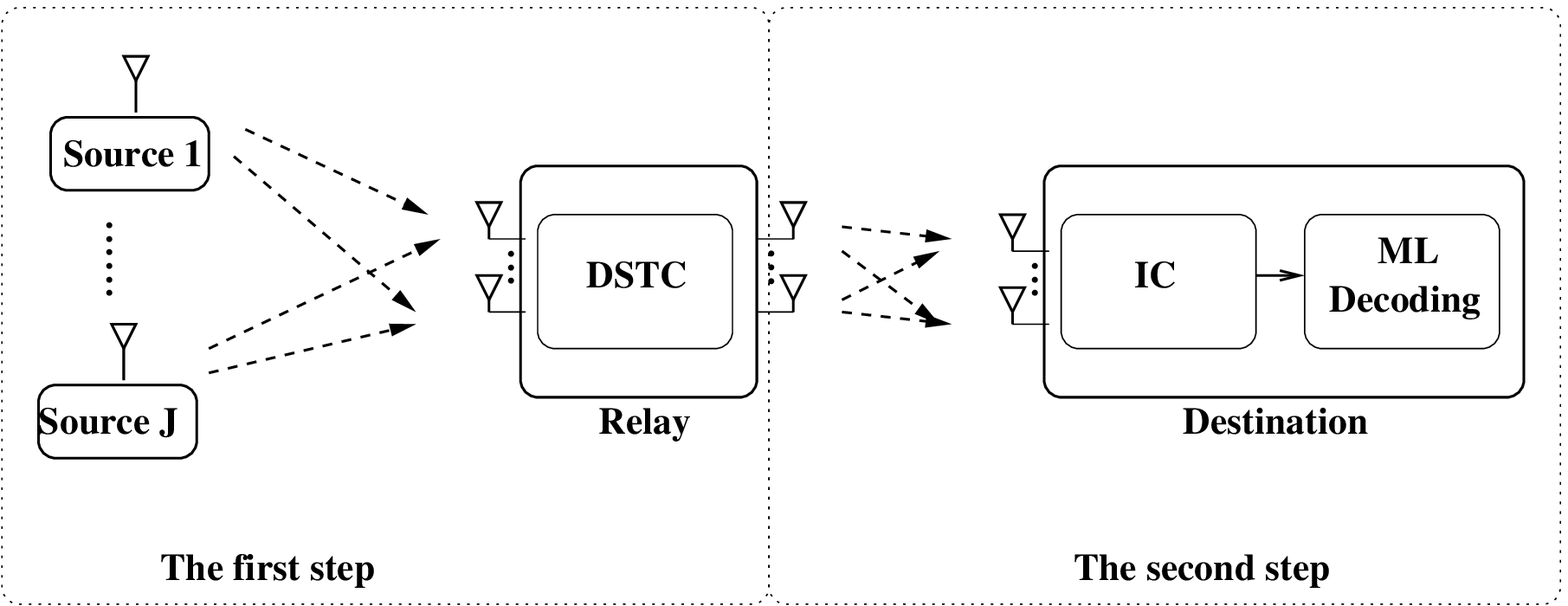}
\caption{System block diagram of \DSTC-ICRec.}
\label{fig-DSTCICRec-block}
\end{figure}
\begin{figure}[!ht]
\centering
\includegraphics[width=4.5in]{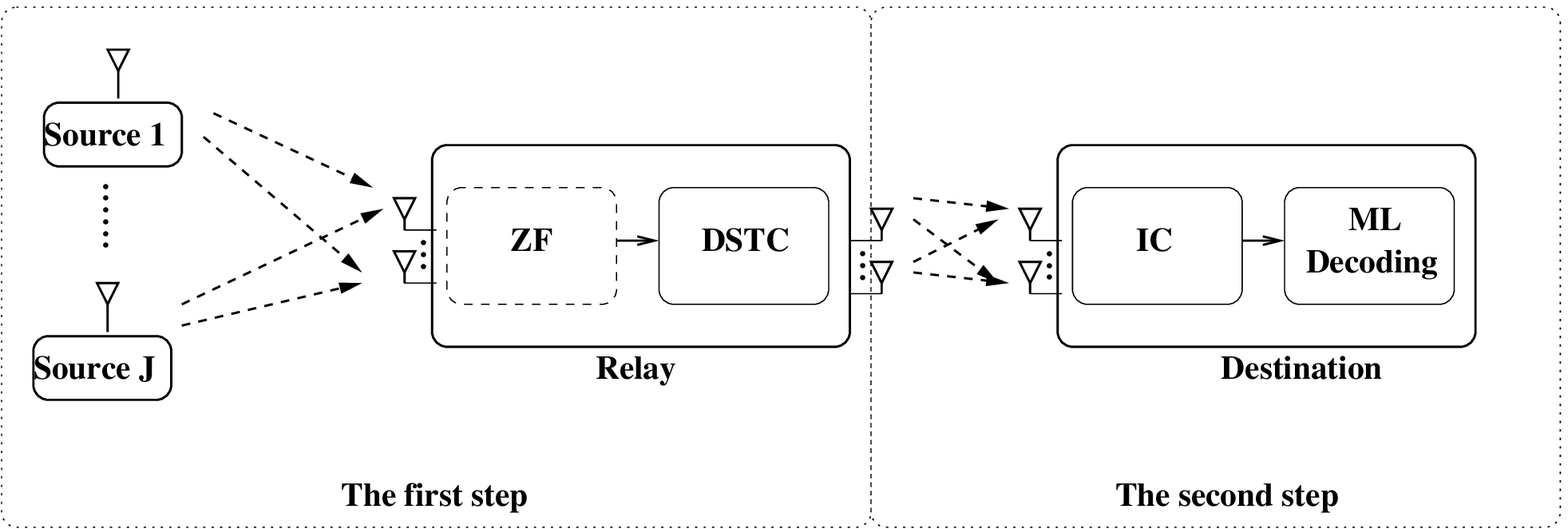}
\caption{Equivalent system of \DSTC-ICRec with zero-forcing at the
relay.} \label{fig-eqsystem}
\end{figure}
%
%
\begin{figure}[!ht]
\centering
\includegraphics[width=4.5in]{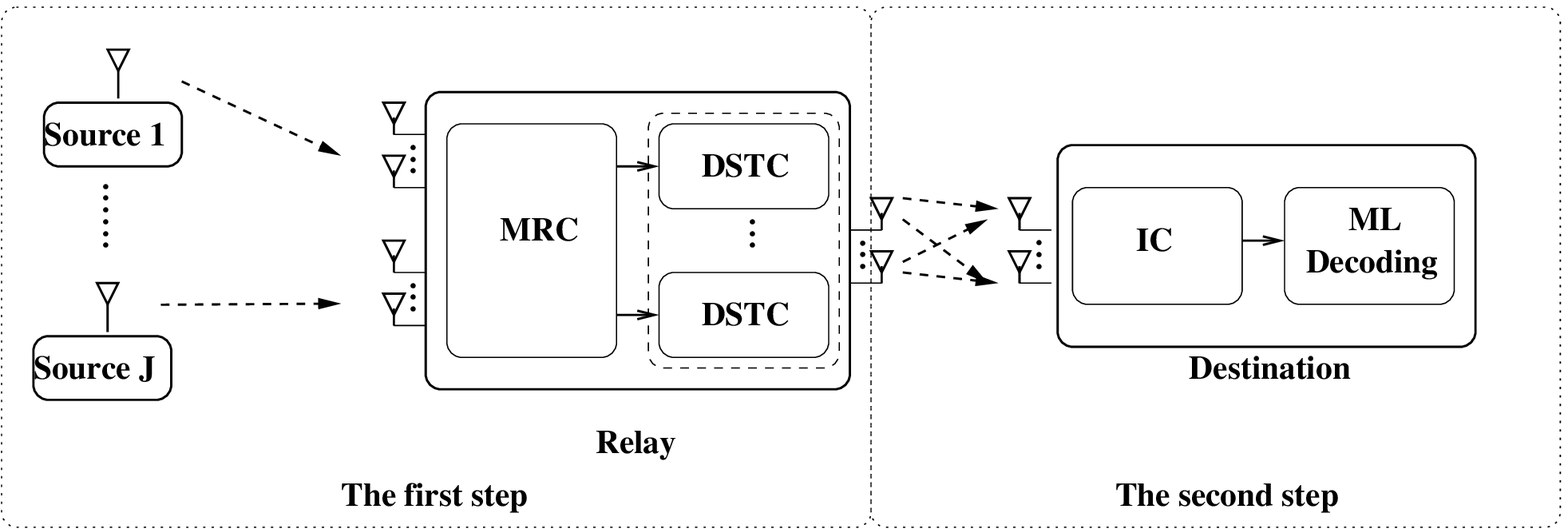}
\caption{System block diagram of \TDMA-ICRec.}
\label{fig-TDMAICRec-block}
\end{figure}
%
%

\begin{figure}[!ht]
\centering
\includegraphics[height=4.5in, angle=-90]{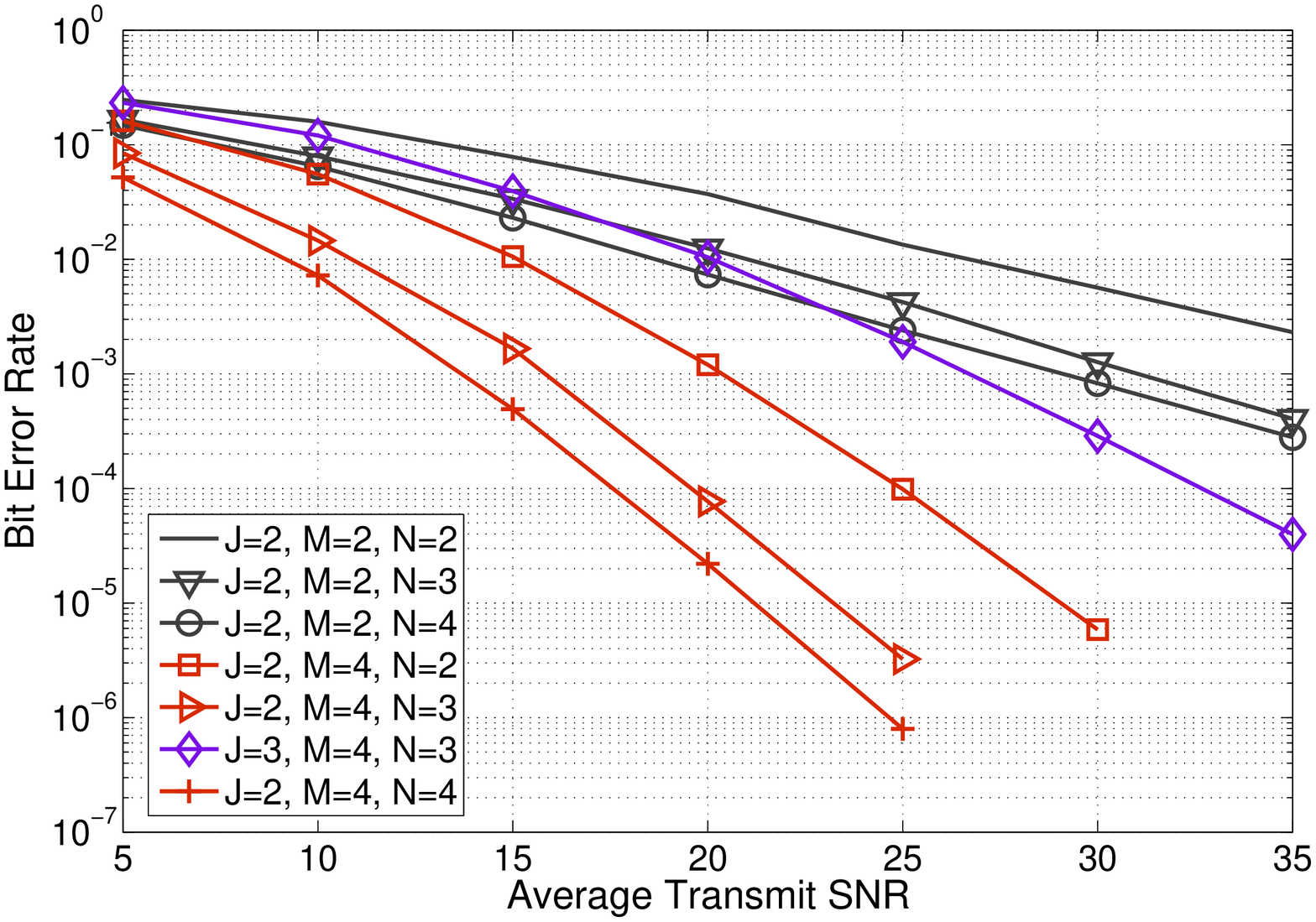}
\caption{BER performance of \DSTC-ICRec, using BPSK modulation.}
\label{fig-DSTCICRec-BER}
\end{figure}

\begin{figure}[!ht]
\centering
\includegraphics[height=4.5in, angle=-90]{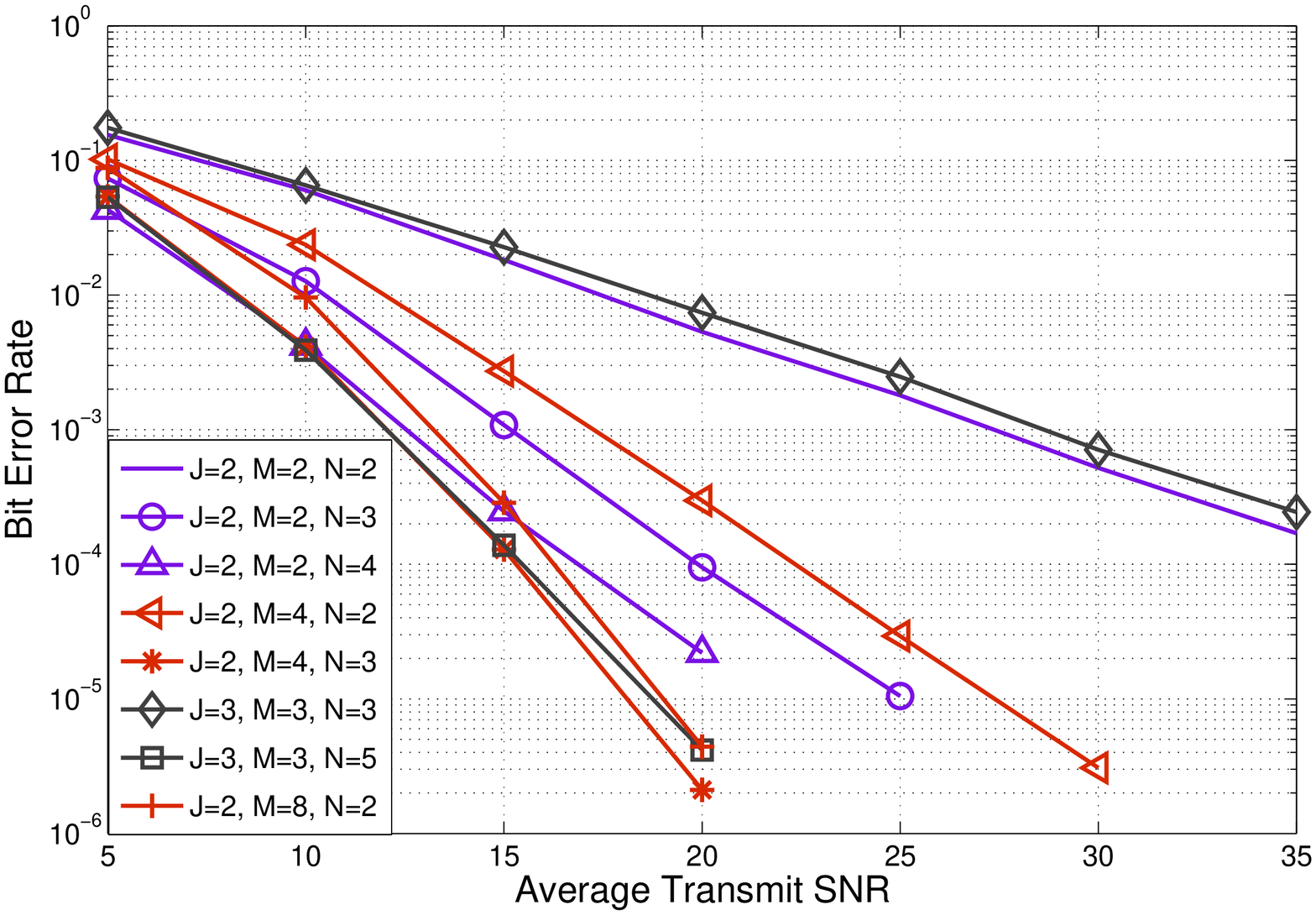}
\caption{BER performance of \TDMA-ICRec, using BPSK modulation.}
\label{fig-TDMAICRec-BER}
\end{figure}

\begin{figure}[!ht]
\centering
\includegraphics[height=4.5in,angle=-90]{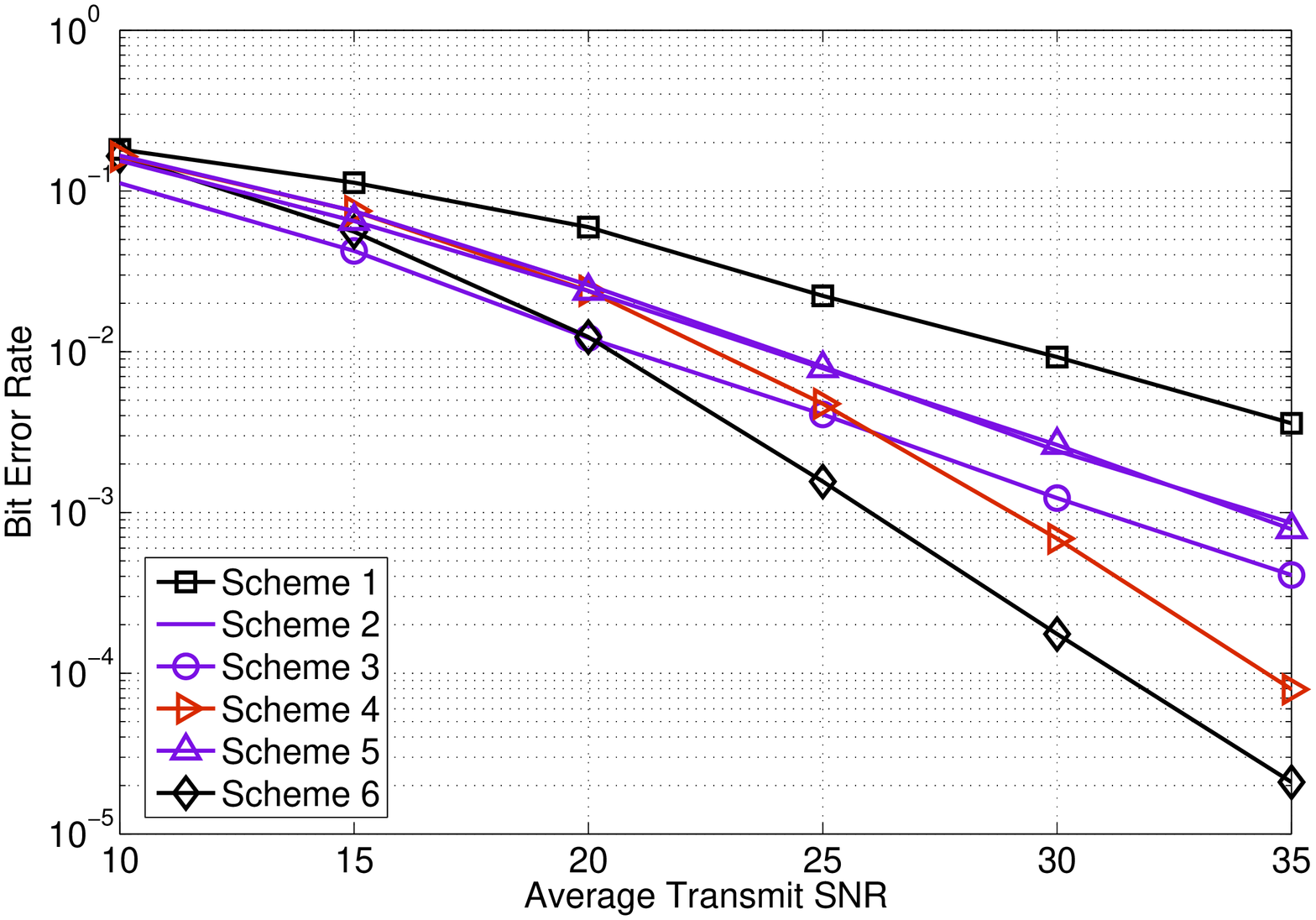}
\caption{Performance comparison in a $1_2\times 2_1\times 2_1$ MARN,
1 bit/source/channel use for all schemes.} \label{fig-comp1}
\end{figure}

\begin{figure}[!ht]
\centering
\includegraphics[height=4.5in,angle=-90]{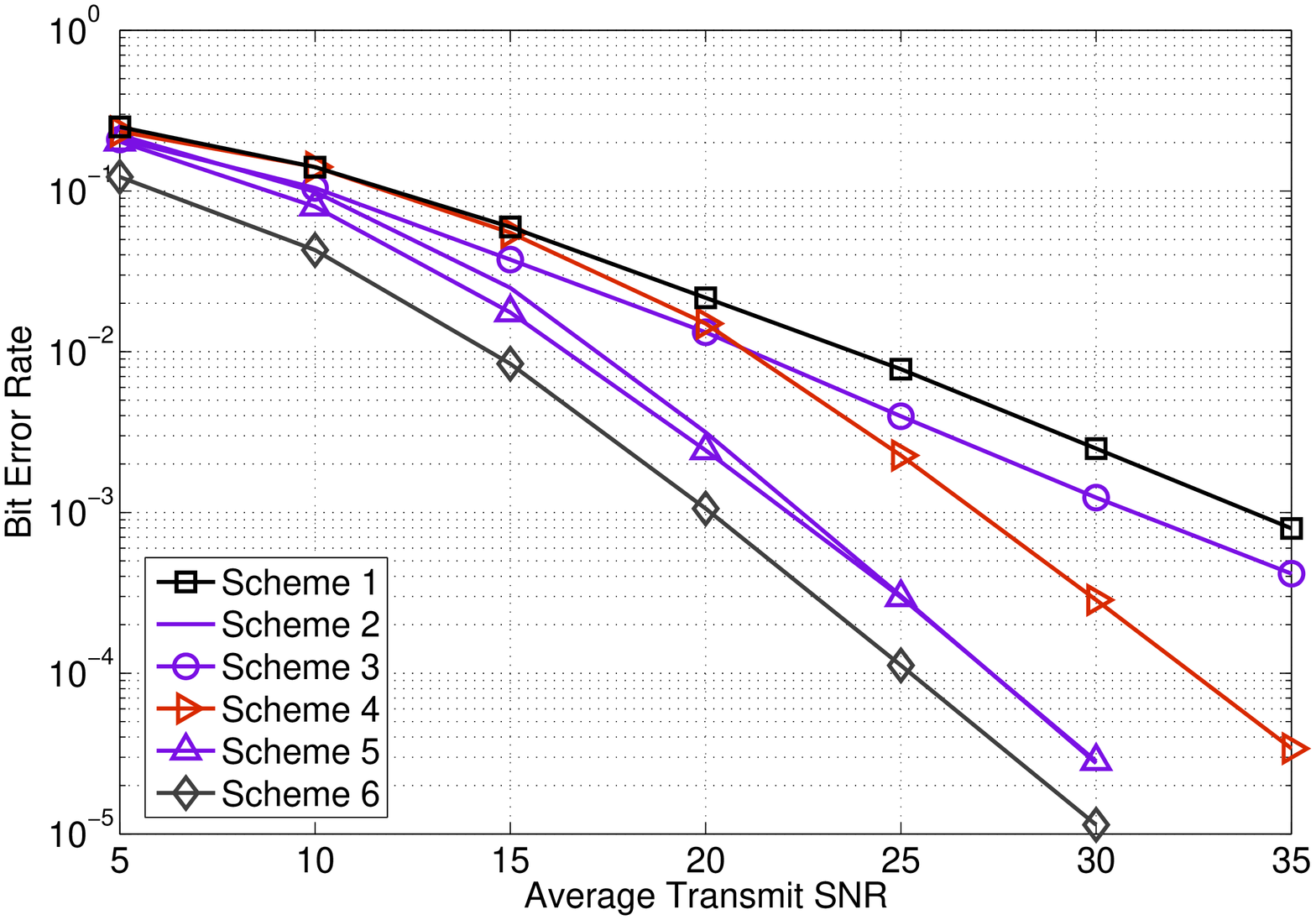}
\caption{Performance comparison in a $1_2\times 2_1\times 3_1$ MARN,
1 bit/source/channel use for all schemes.} \label{fig-comp2}
\end{figure}

\begin{figure}[!ht]
\centering
\includegraphics[height=4.5in,angle=-90]{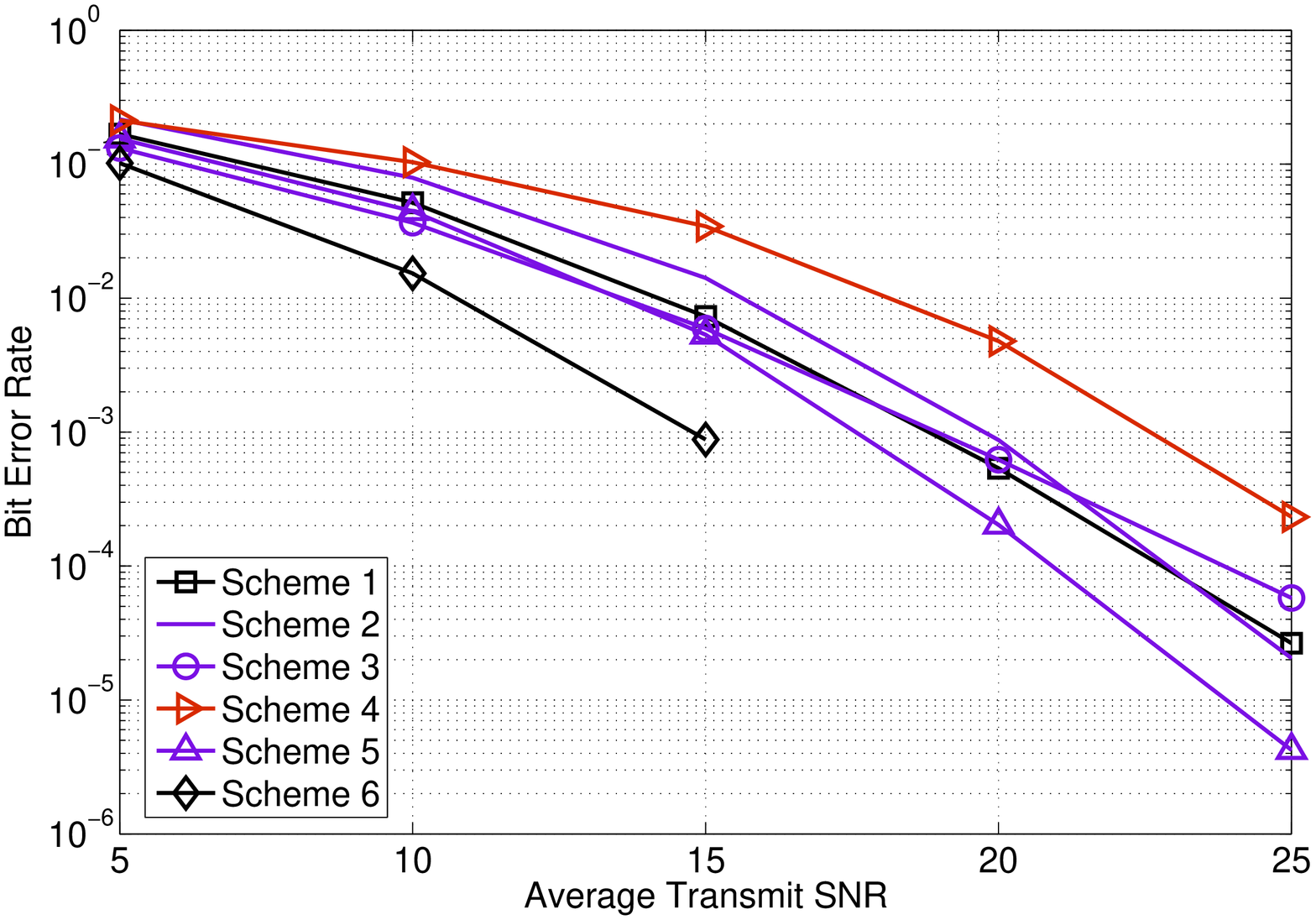}
\caption{Performance comparison in a $1_2\times 4_1\times 3_1$ MARN,
1 bit/source/channel use for all schemes.} \label{fig-comp3}
\end{figure}
\end{document}